\newif\ifconf
\conffalse
\documentclass[letterpaper,twocolumn,10pt]{article}
\ifconf
  \usepackage{usenix-2020-09}
\else
  \usepackage[margin=1in]{geometry}
  \usepackage{mathptmx}
  \usepackage[T1]{fontenc}
  \usepackage[utf8]{inputenc}
  \usepackage[kerning,spacing]{microtype}
  \usepackage{cite}
  \usepackage[hidelinks]{hyperref}
  \usepackage{placeins}
\fi
\usepackage{xurl}
\usepackage{graphicx}
\usepackage{booktabs}
\usepackage{amssymb}
\usepackage{float}
\usepackage{mathtools}
\usepackage{subcaption}
\usepackage{makecell}
\usepackage{amsthm}
\newtheorem{proposition}{Proposition}

\usepackage{cleveref}
\crefname{appendix}{Appendix}{Appendices}
\Crefname{appendix}{Appendix}{Appendices}
\crefname{proposition}{Proposition}{Propositions}
\Crefname{proposition}{Proposition}{Propositions}
\crefname{definition}{Definition}{Definitions}
\Crefname{definition}{Definition}{Definitions}

\newcommand{\papertitle}{Removing the Watermark Is Not Enough: Forensic
Stealth in Generative-AI Watermark Removal}
\ifconf
  \title{\Large\bfseries \papertitle}
  \author{{\rm Anonymous Author(s)}}
\else
  \title{\papertitle}
  \author{
    Yevin Nikhel Goonatilake\\
    George Mason University\\
    \texttt{ygoonati@gmu.edu}
    \and
    Giuseppe Ateniese\\
    George Mason University\\
    \texttt{ateniese@gmu.edu}
  }
\fi
\date{}
\ifconf\else
  \hypersetup{
    pdftitle={\papertitle},
    pdfauthor={Yevin Nikhel Goonatilake and Giuseppe Ateniese},
    hypertexnames=false
  }
\fi

\begin{document}
\maketitle
\ifconf\else
  \pdfbookmark[1]{Abstract}{abstract}
\fi
\begin{abstract}
The literature on watermark removal has largely asked whether an attacker
can make the watermark verifier fail while preserving the appearance of
the image. This is a useful test, but it does not capture the purpose of
removal in applications where watermarks support provenance. In such
settings, the attacker wants the image to pass as ordinary content. If
the removal process leaves a recognizable statistical trace, the
watermark may have disappeared, but deniability has not been restored.
We call this missing requirement forensic stealth.

We evaluate six recent attacks spanning four different removal
strategies and find that all leave strong forensic traces. At a $1\%$
false-positive target, attack-specific detectors identify at least
$99\%$ of the removal outputs. In a separate image-by-image assessment
of five attacks, only one of 750 outputs removes the watermark, remains
within the fidelity budget, and evades forensic detection. The
consistency of this result across different mechanisms shows that
current evaluation practice overlooks a central part of the security
problem.

We also ask whether forensic stealth is possible in principle. Under
explicit idealized assumptions, we show that exact forensic stealth is
possible when a remover preserves source content and resamples the
remaining detail from the corresponding clean distribution. In this
model, the resulting outputs exactly match the clean-image distribution
while remaining within the distortion budget. This shows
that removal traces are not inevitable and places the practical
difficulty in generating source-appropriate clean variation without
damaging the image. We argue that forensic stealth should become part
of the standard by which watermark removal is judged.
\end{abstract}

\section{Introduction}

Generative-AI systems have made it easy to create persuasive synthetic
images, including deepfakes, fabricated campaign material, and other
content meant to mislead. This turns provenance into a practical
security problem. Platforms, investigators, journalists, and end users
need mechanisms that help them decide when an image should be treated as
machine-generated or otherwise suspicious. That motivation has appeared not only in technical work but also in policy: a 2023 U.S. executive order and the EU AI Act
both identified watermarking and related provenance mechanisms as
responses to risks from AI-generated content~\cite{US23,EU24}.

Watermarking is one of the main technical proposals for addressing this
problem~\cite{US23,EU24,sok_watermarking_aigc}. The idea is to embed an
imperceptible marker into model outputs during or after generation time and to allow
a later verifier to test for that marker. For such a mechanism to be
useful, a watermark should satisfy several conditions at once. It
should not visibly harm the image, it should be trustworthy with a low false
positive rate, it should remain detectable after
routine editing and moderate adversarial manipulation, and in the
cryptographic setting it should remain undetectable to parties who do
not know the verification
key~\cite{Aaronson22,Kirchenbauer_23,sok_watermarking_aigc,C:ChrGun24,gunn2025undetectable}. Recent work has made these goals precise for generative models by
casting watermarking in coding-theoretic
terms~\cite{C:ChrGun24,gunn2025undetectable,francati2025codinglimits}.

Recent theory has also clarified that robustness has real limits. The
``Watermarks in the Sand'' (WiTS)
work~\cite{zhang2024watermarks} shows that strong
robustness does not preclude removal: an adversary with enough power,
and in particular access to quality and perturbation oracles, can drive
the watermark below detection while preserving content quality. Yet, this
does not make watermarking pointless. In
many deployment environments, watermarking can still be useful as a
screening signal, a moderation aid, or a provenance cue against less
capable attackers~\cite{sok_watermarking_aigc}. But this perspective does change how watermark
removal should be evaluated. The relevant question is whether
defeating the verifier also restores deniability in practice.

That is where the current removal literature misses the point. Most removal evaluations treat a remover as successful when two conditions are met:
the original watermark detector no longer triggers, and the output image
still looks good~\cite{unmarker_sp25,watermarkattacker_neurips24,ctrlregen_iclr25,boundary_leakage,nfpa_attack,zhang2024watermarks}.
Those conditions are natural, but they do not capture the attacker's
real objective. In the settings that motivate watermark removal, success
means restoring deniability. If the image can still be flagged as
manipulated because the removal process itself left a recognizable
trace, the attack remains unsuccessful. Concurrent work 
independently identifies this forensic leakage as a third evaluation axis~\cite{forensic_cost_arxiv26}.

Consider a simple example. Suppose an attacker generates a fake
political image, removes its watermark with UnMarker~\cite{unmarker_sp25},
and then distributes the result as ordinary content. If a downstream detector
can still identify it as the output of a removal pipeline, the image may remain
subject to further scrutiny. This missing property is what
we call \emph{forensic stealth}: after removal, the output should not
only look visually acceptable, but should also be statistically
indistinguishable from ordinary non-removal releases.

In our experiments, current watermark removers do not return images to
a clean forensic state. Across six attacks spanning
four removal families, detectors trained independently for each remover
distinguish processed outputs from clean images at $99\%$ or
better true-positive rate at a validation-calibrated $1\%$ false-positive
target, while a pooled detector evaluated on the combined six-set mixture
reaches $97.62\%$.
Across several distinct mechanisms, current
removal pipelines consistently introduce a separate detectable signal,
regardless of whether the original watermark remains detectable.

To test whether this is a general weakness of the current attack
landscape, we study the families represented in the WAVES
benchmark~\cite{waves_arxiv24} and extend beyond them. Our evaluation
includes distortion-based optimization (UnMarker~\cite{unmarker_sp25}),
diffusion-based regeneration
(WatermarkAttacker (WMA)~\cite{watermarkattacker_neurips24},
CtrlRegen+~\cite{ctrlregen_iclr25}), latent-space inversion and
perturbation (the Next Frame Prediction Attack (NFPA)~\cite{nfpa_attack},
Boundary Leakage~\cite{boundary_leakage}), and stochastic erosion
(Watermarks in the Sand~\cite{zhang2024watermarks}). Across all six
attacks we observe learnable forensic traces.  We then analyze UnMarker
in depth to characterize the phenomenon, showing that its trace
persists under various post-processing, has a distinctive residual-spectrum profile, and creates a three-way tension among removal success,
image quality, and forensic stealth. This changes the design target for
watermark removal. A next-generation remover will have to meet a
substantially stronger objective: it must erase the watermark,
preserve the utility of the image, and return the output to a clean
forensic state.

Our contributions are the following:

\noindent\textbf{(1). A security formulation of forensic stealth.}
We treat deniability, rather than evasion of a particular detector, as
the attacker's security objective. Concurrent work by Evennou and
Kijak~\cite{forensic_cost_arxiv26} independently reaches the same
high-level conclusion that forensic detectability is a third axis of
watermark-removal evaluation. We develop this shared insight in the keyed PRC setting by defining forensic stealth as computational
indistinguishability between the remover's released outputs and a
clean-release distribution. We pair this formal target with an exact
same-image empirical criterion: watermark rejection, source fidelity,
and forensic-detector evasion must hold simultaneously on the released image.

\noindent\textbf{(2). Failure of current removal attacks.}
We evaluate six state-of-the-art removers spanning distortion-based
optimization, diffusion-based regeneration, latent-space inversion,
and stochastic erosion. Independent forensic detectors distinguish
processed outputs from clean images at $99\%$ or better true-positive
rate at a $1\%$ false-positive target, and a single pooled detector
reaches $97.62\%$. We also analyze forensic signatures using residual
spectra and learned representations, revealing substantial structure
across the evaluated output sets. In the exact same-image evaluation,
only $1/750$ outputs across five attack-specific samples satisfy
watermark removal, source fidelity, and detector evasion simultaneously
(\Cref{sec:results:joint}).

\noindent\textbf{(3). Validation and adaptive evaluation.} We test whether detection can be explained by metadata, encoding, file size, split leakage, or generic generative processing, and include common image transformations in the clean class to control for incidental processing artifacts (\Cref{sec:background,sec:results:validations}). We additionally show that fixed-detector evasion does not establish forensic stealth in the tested CtrlRegen+ setting: Projected Gradient Descent (PGD) evades the detector's $0.5$ decision threshold on $99.97\%$ of outputs, yet a newly trained detector achieves AUROC~$1.000$ (\Cref{app:adaptive}).

\noindent\textbf{(4). A design target and possibility boundary for
future removers.}
We formulate a three-term objective combining source fidelity,
watermark suppression, and resemblance to ordinary clean releases.
An abstract, idealized construction shows that, under strong assumptions, a remover can reproduce the clean distribution while remaining inside the
distortion budget. Together, these results identify matching the clean
output distribution under a quality constraint as the central challenge
for future removers (\Cref{sec:discussion_stealth,sec:stealth:boundary}).

\section{Background}
\label{sec:background}

\subsection{Forensic baseline}
Current evaluations of watermark removal usually ask whether a watermark can be erased
without visibly damaging the image~\cite{waves_arxiv24,forensic_cost_arxiv26}. For the setting studied here, the
question has a third part. A practically successful remover must defeat
the verifier, preserve the image, and avoid leaving behind a detectable
trace of removal. To make that point precise, we need only a minimal
model of watermarking.

Starting from an image $x$ and a secret key $k$, the embedder produces
a watermarked image $x_w = \mathsf{WM}(x; k)$, and the verifier
$\mathsf{Verify}(\cdot; k)$ decides whether a candidate image should be
treated as watermarked, typically by computing a detection statistic and
comparing it to a threshold. Such a
scheme must satisfy competing demands. The watermark should survive
compression, resizing, and moderate adversarial editing while the image
stays visually close to the original. In a secret-key setting the
verifier can test for it and an observer without the key
cannot~\cite{sok_watermarking_aigc,C:ChrGun24,gunn2025undetectable}.

We work in the pseudorandom-code (PRC) setting introduced by Christ and
Gunn and instantiated for images by Gunn, Zhao, and
Song~\cite{C:ChrGun24,gunn2025undetectable}, where watermarked outputs
are intended to be computationally indistinguishable from clean outputs
to any efficient observer without the secret key, including across
polynomially many samples. For the PRC setting studied here, the
corresponding clean distribution is therefore the right forensic
baseline. Computational indistinguishability is preserved by any
efficient, randomized, key-independent remover $K$: otherwise, composing
a distinguisher with $K$ would break PRC security. Thus, for corresponding
input ensembles, no efficient image-only test can tell whether $K$ was
applied to clean or PRC-watermarked inputs. This guarantee does not cover
unrelated input sources. Results on our heterogeneous generated, natural,
and artistic datasets therefore measure empirical transfer.

This comparison is meaningful only if representation artifacts are kept
out of the signal. A detector can appear to succeed because of file
format, metadata, encoder quirks, compression settings, or systematic
file-size differences introduced by the evaluation pipeline. We
therefore standardize export and preprocessing across both classes so
that any remaining separability reflects the transform itself. We
instantiate these controls in
\Cref{sec:results:validations}.
One shortcut is specific to the five attacks with generative stages. Their
attacked images pass through diffusion inversion, regeneration, or
inpainting, while pristine clean images do not. A detector could
therefore separate the classes by recognizing generative processing
rather than removal. We test this possibility with family-specific matched-light controls, which are lighter than the full pipelines and do not isolate the stage responsible for detection (\Cref{sec:setup:data}).

\subsection{Classes of watermark removal attacks}
Starting from the WAVES taxonomy~\cite{waves_arxiv24}, the removers
studied here fall into four classes. WAVES distinguishes distortion,
regeneration, and adversarial attacks. In the keyed PRC setting
considered here, it is useful to separate latent-space inversion from
regeneration and to treat stochastic erosion as its own class. The four
classes differ in where the attack intervenes: directly in pixel space,
through reconstruction under a generative prior, through perturbation
of an inferred latent state, or through repeated oracle-guided local
edits.

\paragraph{Distortion-based optimization.}
Distortion-based optimization attacks work directly in image space. At the weak end, this
family includes classical transformations such as recompression, blur,
and Gaussian noise. At the stronger end, it includes optimization-based
pipelines that perturb the image directly so as to disrupt watermark
carriers while preserving perceptual quality. We represent this class
with UnMarker~\cite{unmarker_sp25}, a universal black-box pipeline that
operates directly on the image and requires no detector queries.

\paragraph{Regeneration.}
Regeneration attacks remove the watermark by destroying fine image
structure and reconstructing the result through a denoiser or
generative prior. Here the watermark disappears as part of the
reconstruction process rather than because a visible carrier is
directly targeted. WatermarkAttacker~\cite{watermarkattacker_neurips24}
does this in a diffusion-purification style, adding noise and then
denoising so that the watermark is lost during reconstruction.
CtrlRegen+~\cite{ctrlregen_iclr25} instead uses controlled regeneration
with semantic and spatial guidance, rebuilding the image from
progressively noisier latent states while trying to preserve the input
content.

\paragraph{Latent-space inversion.}
Latent-space inversion attacks first invert the image into the latent
representation of a diffusion model, typically through denoising diffusion implicit model (DDIM)~\cite{ddim_iclr21} inversion,
and then perturb that representation before reconstructing the image.
NFPA~\cite{nfpa_attack} does this through calibrated latent
perturbations, while Boundary Leakage~\cite{boundary_leakage} exploits
information exposed by the watermark's detection boundary. Because the
inversion step is tied to the underlying model, these attacks are
limited in our evaluation to images generated by Stable
Diffusion~2.1~\cite{stable_diffusion_21,ldm_cvpr22}.

\paragraph{Stochastic erosion.}
Stochastic-erosion attacks do not target a specific carrier at all.
Instead, they apply many small local edits, often inpainting-style
patches, and retain only those accepted by a quality oracle. The image
then drifts away from the watermarked version through a sequence of
quality-preserving steps, and the watermark degrades as a byproduct of
that drift. We include Watermarks in the
Sand~\cite{zhang2024watermarks} as the prominent example of this class.
Its significance is primarily theoretical: Zhang et al.\ use it as the
constructive attack in their impossibility result, showing that under
access to a perturbation oracle and a quality oracle, strong
watermarking can be defeated by a quality-preserving random walk.

We do not separately evaluate the full adversarial suite in WAVES~\cite{waves_arxiv24}. We treat its surrogate-detector attacks as out of scope because they learn an attack objective from watermark-labeled examples. Its black-box embedding attacks instead optimize public image encoders without access to the watermark signal. These attacks remain meaningful under key secrecy and fall within the broader image-space optimization class represented here by UnMarker.

\subsection{Threat model and success criteria}
In this paper, we define a successful removal attack as one that defeats the verifier, preserves the utility of the image, and returns the output to a clean forensic state. We now make those three conditions precise.

The attacker applies a removal transform $K$ to a watermarked image and
releases only the final output
\[
x_a = K(x_w).
\]
When we study post-processing separately, we instead define the final
release as $x_a=A(K(x_w))$, where $A$ is the post-processing operation;
all three success criteria are evaluated on this output. The defender sees only the released image and asks whether it passed through a removal pipeline. The task is therefore forensic: the defender need not prove that a watermark remains. It may obtain labeled samples
from the specified remover, but has no access to the attacker's model,
prompts, or key, as formalized below.

The first condition is verifier evasion:
\[
\mathsf{Verify}(x_a; k) = 0.
\]
The second is source-fidelity preservation:
\[
d(x_a,x_w)\leq\tau,
\]
where $d(\cdot,\cdot)$ measures the cumulative change from the
watermarked input to the final released output, and $\tau$ is an
application-level tolerance. We instantiate this criterion with LPIPS
in \Cref{sec:setup:quality}. These are the two criteria that
current removal work most often makes explicit~\cite{waves_arxiv24,unmarker_sp25,
watermarkattacker_neurips24,ctrlregen_iclr25,zhang2024watermarks}.

The missing third condition is forensic stealth. We formalize it through
a known-remover distinguishing experiment. Fix a security parameter
$\lambda$ and an efficient, randomized, key-independent remover $K$, and
sample $k\leftarrow\mathsf{KeyGen}(1^\lambda)$ once. Let
$\mathbb{P}_{\mathrm{clean},\lambda}$ and
$\mathbb{P}_{\mathrm{wm},\lambda,k}$ denote the corresponding clean and
watermarked release distributions, and set
$\mathbb{P}_{\mathrm{rem},\lambda,k}^{K}
=K(\mathbb{P}_{\mathrm{wm},\lambda,k})$. The clean and watermarked
distributions include source, prompt, and generation randomness; the
removed distribution also includes the randomness of $K$. The defender
$D$ receives polynomially many labeled samples from the clean and
removed distributions. For a fresh challenge, the challenger samples
$b\leftarrow\{0,1\}$, gives $D$ a sample from
$\mathbb{P}_{\mathrm{clean},\lambda}$ if $b=0$ and from
$\mathbb{P}_{\mathrm{rem},\lambda,k}^{K}$ if $b=1$, and $D$ returns
$b'$. Its advantage is
\[
\operatorname{Adv}^{\mathrm{fs}}_{K,D}(\lambda)
=2\left|\Pr[b'=b]-\frac{1}{2}\right|,
\]
where the probability is over all randomness in the experiment. We call
$K$ forensically stealthy if this advantage is negligible for every
efficient image-only $D$ with this sample access. We use
$\mathbb{P}_{\mathrm{rem}}^K\approx_{\mathrm{c}}
\mathbb{P}_{\mathrm{clean}}$ as shorthand for this experiment. Our experiments instantiate this known-remover game on specified clean and removed distributions; they do not establish transfer to an unseen remover or a different clean distribution.

This distributional condition is the security target; our per-image
endpoint is an empirical test against one fixed detector. A successful
distinguisher is evidence against stealth, whereas evading this detector
does not establish indistinguishability. Fix the detector and choose a
threshold $t$ on an independent clean validation split; let $D_t$ denote the
resulting binary test. A single release succeeds only if all three
conditions hold on that same image:
\[
J_K(x_w)=\mathbf{1}\!\left[
\begin{aligned}
&\mathsf{Verify}(x_a;k)=0 \;\wedge\;\\
&d(x_a,x_w)\le\tau \;\wedge\;\\
&D_t(x_a)=0
\end{aligned}
\right].
\]

The joint success rate is the mean of $J_K$ over common image
identifiers. It is not the product of the three marginal rates. A remover
that removes on one subset and preserves quality on a disjoint subset has
a positive product and a joint of zero. We therefore report the
conjunction directly (\Cref{sec:results:joint}).

\section{Experimental Setup}
\label{sec:setup}

For each remover $K$, we ask whether one can tell from the final image
alone that the image has passed through the removal pipeline. We therefore treat $K$ as the randomized final-output transform induced by a
remover. Starting from a shared
image pool, we form disjoint clean and attacked classes, optionally apply the same family of
lightweight post-processing operations to both classes, and then train
a detector specific to $K$ to predict whether that transform was
applied (see \Cref{fig:pipeline}). NFPA, Boundary Leakage, and WiTS use
attack-specific dataset constructions, described below, because of input
compatibility, code availability, and computational constraints. Our primary evaluation trains a separate detector for each remover, so each result asks whether that pipeline is forensically detectable on its own terms. We also train a pooled detector to test whether one model can detect traces across the six evaluated removal methods. If a detector succeeds under controlled conditions, the corresponding
remover fails the empirical forensic-stealth test. \Cref{sec:results:validations}
examines shortcut-artifact controls across the six detectors.

\Cref{tab:unified-detection} reports
watermark removal, paired source fidelity, and forensic detection. The
evaluation uses two types of inputs. For detector and signature analyses,
we apply $K$ to both watermarked and never-watermarked inputs. Under the closure argument in
\Cref{sec:background}, this is a valid computational proxy when those
inputs belong to the corresponding clean distribution. The reported verifier outcomes,
paired fidelity measurements, and exact three-way results use genuinely
watermarked inputs directly. The LPIPS entries use separate paired
cohorts of $100$ outputs per attack. For Boundary Leakage, we report
forensic detection and the authors' removal rate, but no result from the
selected LPIPS audit. We evaluate exact three-way success separately on
five attack-specific cohorts of $150$ outputs each
(\Cref{tab:joint-outcome}).

\begin{figure}[t]
  \centering
  \includegraphics[width=\linewidth]{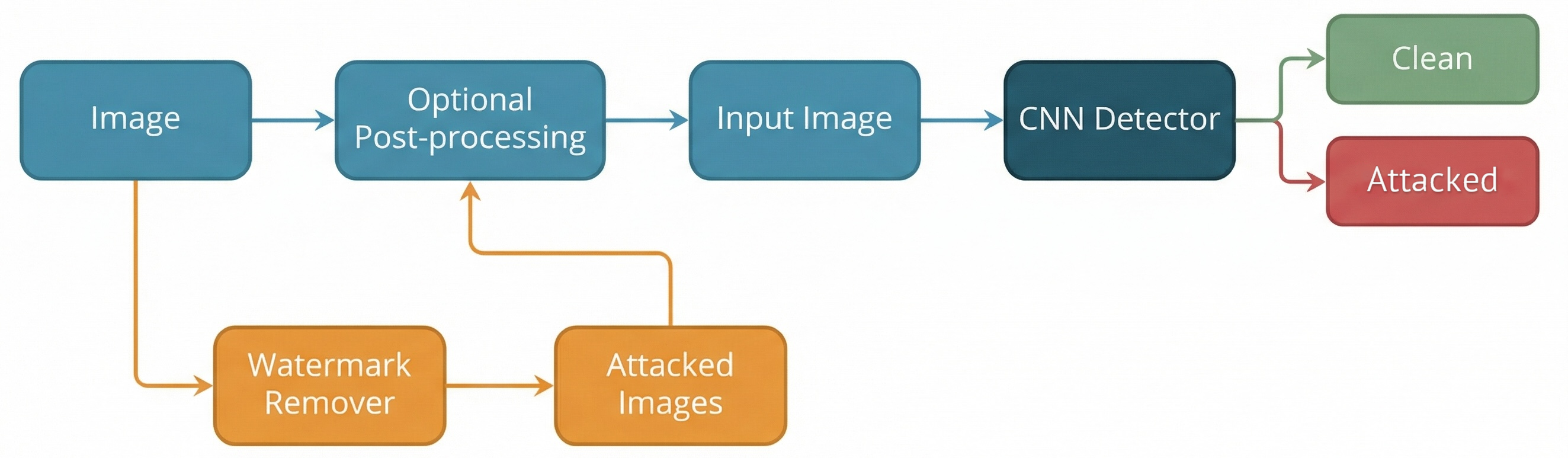}
  \caption{%
    \textbf{Construction and detection pipeline.}
    Clean releases form the negative class and removal outputs form the
    attacked class. Both may be post-processed before the detector
    decides whether an image passed through the pipeline.}
  \label{fig:pipeline}
\end{figure}

\subsection{Datasets}
\label{sec:setup:data}

\paragraph{Image sources.}
\label{sec:setup:datasources}
All images are standardized to $512 \times 512$ RGB. The common base pool
contains $157{,}984$ images drawn from five sources chosen to
span natural, online, photographic, artistic, and model-generated
content. U512 is our local name for a mixed-domain real-world corpus
drawn from the 130k Images (512x512) - Universal Image Embeddings
dataset~\cite{u512_kaggle}. AbstractArt and ArtMix are our local names
for corpora drawn from Abstract Art Images~\cite{abstractart_kaggle}
and Art Images Clear and Distorted~\cite{artmix_kaggle}, respectively;
ArtMix is included to test heterogeneity in sharpness and deformation.
Caltech256 contributes object-centric photographs from the
Caltech-256 dataset~\cite{caltech256_tr}. PRC-Gen contains images that
we generated under a PRC-based watermarking pipeline from prompts drawn
from the Gustavosta prompt dataset~\cite{gustavosta_sd_prompts}. These short names are local shorthand for the fixed corpora used
throughout the evaluation.

\paragraph{Common construction.}
From this common pool, we first separate underlying images into two
disjoint halves. One half provides ordinary clean releases for the
negative class. The other half provides inputs to the corresponding
removal pipeline, whose final outputs form the attacked class.
Formally, let $\mathcal{D}_{\mathrm{base}}$ denote the common image pool. We reserve one half as the clean class
$\mathcal{D}_{\mathrm{clean}}$ and the other as the attack-reserved half
$\mathcal{D}_{\mathrm{base,atk}}$, each of size $78{,}992$, and define
\[
\mathcal{D}_{\mathrm{atk}} = \{K(x) : x \in \mathcal{D}_{\mathrm{base,atk}}\}.
\]
The clean and attacked classes are therefore disjoint at the level of
underlying images. To test robustness to post-processing, we
define a family $\mathcal{A}$ of ten lightweight editing operations:
crop-resize, rotation, scaling, JPEG recompression, chroma subsampling,
quantization, Gaussian blur, bilateral filtering, non-local means
denoising, and color jitter. We then sample $20{,}000$ images from each
class and apply a random $A \sim \mathcal{A}$ to produce $40{,}000$
tampered images, split equally between tampered-clean and
tampered-attacked examples. The binary label is
\[
y(z) =
\begin{cases}
1 & \text{if } z \in \mathcal{D}_{\mathrm{atk}} \text{ or } z
    \text{ is tampered-attacked},\\
0 & \text{if } z \in \mathcal{D}_{\mathrm{clean}} \text{ or } z
    \text{ is tampered-clean}.
\end{cases}
\]UnMarker~\cite{unmarker_sp25},
WatermarkAttacker~\cite{watermarkattacker_neurips24}, and
CtrlRegen+~\cite{ctrlregen_iclr25} use this common construction,
yielding $197{,}984$ images after tampered augmentation.

\paragraph{Attack-specific datasets.}
The common construction is used directly by only those three attacks. NFPA~\cite{nfpa_attack} and Boundary Leakage~\cite{boundary_leakage} instead require SD2.1-compatible inputs, while WiTS~\cite{zhang2024watermarks} is
evaluated on a subset because its sequential random walk is expensive. Exact per-attack dataset sizes are reported in \Cref{tab:per-attack-data} of \Cref{app:training-details}.

\emph{NFPA.} NFPA~\cite{nfpa_attack} operates through DDIM~\cite{ddim_iclr21}
inversion and therefore requires inputs generated by a model compatible
with Stable Diffusion~2.1~\cite{stable_diffusion_21,ldm_cvpr22}. Its
input pool combines $14{,}997$ PRC-Gen images with $10{,}000$ additional
SD2.1 images generated from the same Gustavosta prompt dataset
\cite{gustavosta_sd_prompts} and settings. After the attack-specific validity
filter, $24{,}996$ images remain per class, for a $49{,}992$-image base
dataset and $69{,}992$ images after tampered augmentation.

\emph{Boundary Leakage.} Boundary Leakage~\cite{boundary_leakage} is
also restricted to SD2.1-generated inputs, and the attack code was not
publicly available at the time of evaluation. We therefore use the
authors' released set of approximately $10{,}000$ attacked images
together with an independently generated SD2.1 clean set and
$3{,}000$ tampered images, for a total of approximately $23{,}000$
images. Our reference follows the same published prompt
dataset, model family, resolution, and number of diffusion steps,
reducing potential distribution mismatch. Because the attack itself was 
unavailable, we report Boundary Leakage as a comparison against released 
outputs rather than a pipeline we ran ourselves. This limitation does not 
change the result itself: the released attacked images remain readily 
distinguishable from the matched clean reference.

\emph{WiTS.} Watermarks in the Sand~\cite{zhang2024watermarks} applies a
quality-preserving random walk with 50 sequential inpainting passes per
image. This makes the attack computationally expensive and limits
throughput, so we evaluate it on a subset of the common pool, yielding
approximately $14{,}568$ images per class and approximately $32{,}136$
images total after tampered augmentation.

\paragraph{Matched-light clean negatives.}
Several attacks use generation, inversion, or inpainting, so their
detectors may respond to processing strength rather than a
removal-specific trace. We evaluate benign images that undergo lighter,
attack-specific versions of these operations and call them matched-light
negatives. Because they are not matched to the full pipelines in
processing strength, these controls test only whether light processing is sufficient to trigger the detectors; they do not isolate the stage responsible for detection. These controls are labeled clean, so a detector firing on them constitutes a false positive rather than evidence of watermark removal. UnMarker has no generative stage and therefore no matched-light counterpart. \Cref{tab:matched-light} lists each control.

\begin{table}[t]
\centering
\ifconf\small\else\footnotesize\fi
  \caption{\textbf{Benign controls for the generative-processing shortcut.}
  The control construction is attack-specific; UnMarker has no
  generative stage and therefore no matched-light control.}
  \label{tab:matched-light}
  \begin{tabular*}{\columnwidth}{@{\extracolsep{\fill}}ll@{}}
\toprule
\textbf{Attack} & \textbf{Benign control operation} \\
\midrule
Boundary Leakage & invert--denoise round-trip, $\epsilon = 0$ \\
NFPA             & no frequency perturbation ($\mathrm{nfp}_{xy}=0$) \\
WMA              & low regeneration noise (noise step $20$) \\
WiTS             & inpaint ${\approx}2\%$ of pixels overall \\
CtrlRegen+       & disjoint-source clean regenerations \\
\bottomrule
\end{tabular*}
\end{table}

\paragraph{Splits and integrity.}
\label{sec:setup:splits}
We split each dataset into train, validation, and test using 70/15/15
target proportions at the level of underlying images, stratified by
label and source so that each split preserves the same content-type mixture. Each source image and all of its derived versions therefore remain in the same split. Tampered images are balanced across operators, and the benign negatives were simply a $20\%$ portion of the original clean set. For pipelines that require attack-specific filtering or generated outputs, the final evaluation uses only usable images that pass those filters. 
After splitting, all files are renamed under a uniform convention to prove the detector doesn't exploit filenames, paths, or directory structure. Any learnable signal must therefore come from rendered pixels rather than filenames, paths, or directory layout.

\subsection{Quality and detector metrics}
\label{sec:setup:quality}
\label{sec:setup:detector}

\paragraph{Source fidelity.}
For each executable attack with valid source--output pairs, we evaluate a
fixed sample of $100$ outputs (seed 20260815). We measure cumulative
input-to-output LPIPS~\cite{zhang2018lpips} using AlexNet at
$256\times256$ with RGB inputs mapped to $[-1,1]$. An output passes when
$\mathrm{LPIPS}\leq0.35$. We report the median and pass count separately
from the detector test sets and the $n=150$ joint evaluation. We use
$0.35$ as our evaluation budget, not as a general cutoff for perceptual
similarity.

\paragraph{Forensic detector.}
Each attack receives an independent ResNet-50
detector~\cite{he2016resnet}, initialized from ImageNet-1K
weights~\cite{deng2009imagenet}, fine-tuned end-to-end, and evaluated on native
$512\times512$ images. We also train one pooled detector whose attacked class
combines all six removers; it is reported only as an aggregate result and used
for the representation analysis. Training and hardware details appear in
\Cref{app:training-details,app:pooled-detector}.

The detector outputs a removal score in $[0,1]$. Because the detectors are trained independently, raw scores are interpreted only relative to each detector's frozen threshold and are not compared across attacks. Our primary operating point is chosen on an independent validation split to target $1\%$ FPR and frozen before testing. We also report AUROC and held-out test-ROC TPR at $0.1\%$ FPR, treating the latter as descriptive because of its lower empirical resolution.

\section{Results}
\label{sec:results}

We begin with the broadest question: does forensic detectability appear
across the foremost attack families? Under the protocol of
\Cref{sec:setup}, detectors trained separately for each remover
distinguish processed outputs from clean images across all six attacks. In addition, the pooled detector separates clean and attacked images within the combined held-out mixture. The rest of the section asks whether that result survives scrutiny. We test two ways it could be an artifact: a representation shortcut and a generative-processing shortcut (\Cref{sec:results:validations}). We then ask how much post-processing is needed to weaken the trace (\Cref{sec:results:robustness}). Throughout, we report AUROC as a global separability measure and emphasize TPR@1\%\,FPR and TPR@0.1\%\,FPR, which reflect deployment settings where false alarms are critical.

\subsection{Detection across attack families}
\label{sec:results:detection}

For each removal algorithm, we train an independent forensic detector
on its own clean-vs.-attacked dataset using the protocol of
\Cref{sec:setup}. \Cref{tab:unified-detection} shows the central point
of this subsection: detectability appears across all six evaluated
attacks, spanning multiple removal strategies and datasets. At the validation-frozen operating points, realized held-out FPR ranges
from $0.77\%$ to $1.12\%$ across the six detectors.

\begin{table*}[t]
\centering
\footnotesize
\caption{\textbf{Forensic detection and attack-side outcomes across six
removal attacks.} Each row reports its attack-specific detector; the
final row reports the pooled detector. WM ASR is verifier rejection
among genuinely watermarked inputs, with the evaluated count in
parentheses. Acc is at the default $0.5$ threshold. Note NFPA and WiTS have higher LPIPS because inversion and repeated inpainting
resample fine image detail. $^{a}$Boundary Leakage is the
authors' reported value under their watermark and key.}
\label{tab:unified-detection}
\setlength{\tabcolsep}{3pt}
\begin{tabular*}{\textwidth}{@{\extracolsep{\fill}}llrrrrrr@{}}
\toprule
\textbf{Attack} & \textbf{Family}
  & \textbf{Acc}
  & \textbf{AUROC}
  & \makecell{\textbf{TPR at frozen}\\\textbf{1\% target}}
  & \makecell{\textbf{Test-ROC TPR@}\\\textbf{0.1\% FPR}}
  & \textbf{WM ASR}
  & \makecell{\textbf{LPIPS median}\\\textbf{[pass/$n$]}} \\
\midrule
UnMarker       & Dist.
  & 99.56\% & 0.9998 & 99.92\% & 99.76\%
  & 1.000 (186) & 0.217 [100/100] \\
WMA            & Regen.
  & 99.43\% & 0.9994 & 99.82\% & 94.62\%
  & 0.011 (186) & 0.066 [100/100] \\
CtrlRegen+     & Regen.
  & 99.66\% & 0.9999 & 99.95\% & 99.00\%
  & 0.668 (500) & 0.120 [100/100] \\
NFPA           & Inv./Pert.
  & 99.31\% & 0.9994 & 99.59\% & 93.73\%
  & 1.000 (168) & 0.426 [6/100] \\
Boundary Leak. & Inv./Pert.
  & 99.39\% & 0.9995 & 99.53\% & 98.83\%
  & 0.31$^{a}$ & --- \\
WiTS           & Erosion
  & 99.04\% & 0.9979 & 99.09\% & 96.23\%
  & 1.000 (123) & 0.357 [47/100] \\
\midrule
Pooled detector & All six
  & 98.24\% & 0.9976 & 97.62\% & 82.43\%
  & --- & --- \\
\bottomrule
\end{tabular*}
\end{table*}
 
Across the six evaluated attacks, reported ASR spans 0.011 to 1.000.
They are not uniformly successful removers, and detectability does not
track removal strength: WMA removes the watermark on almost
nothing and is still caught, while CtrlRegen+ removes it on two thirds of
watermarked inputs and is also caught. Because these rates are measured separately, they do not reveal how often a single output satisfies all three conditions. We measure that exact per-image conjunction of removal, source fidelity, and detector evasion in \Cref{sec:results:joint}.

\paragraph{Distortion-based optimization (UnMarker).}
UnMarker~\cite{unmarker_sp25} is the most general-purpose attack in our
evaluation and serves as the primary case study for the in-depth
analyses of \Cref{sec:results:robustness,sec:signature}. It uses iterative, LPIPS-constrained optimization to suppress the image features that carry the watermark. On the held-out test split it reaches AUROC $0.9998$ and holds $99.76\%$ TPR at $0.1\%$\,FPR (\Cref{tab:unified-detection}).

\paragraph{Regeneration-based removal (WatermarkAttacker, CtrlRegen+).}
WatermarkAttacker~\cite{watermarkattacker_neurips24} reconstructs the
image through a diffusion model after calibrated noising.
CtrlRegen+~\cite{ctrlregen_iclr25} performs a related regeneration step
in latent space, but uses control networks to preserve structure, layout,
and color. Their removal rates differ sharply: WatermarkAttacker removes
the watermark from $2/186$ inputs ($1.1\%$), while CtrlRegen+ succeeds on
$334/500$ ($66.8\%$). Despite this difference, both remain easily
detectable, reaching $99.82\%$ and $99.95\%$ TPR, respectively, at
$1\%$ FPR.

\paragraph{Latent-space inversion (Next Frame Prediction Attack, Boundary Leakage).}
NFPA~\cite{nfpa_attack} and Boundary Leakage~\cite{boundary_leakage}
both rely on DDIM inversion, latent-space perturbation, and regeneration
back to pixels. NFPA uses calibrated latent noise, whereas Boundary
Leakage explicitly probes the watermark verifier's decision boundary.
Both attacks require SD2.1-generated inputs and are therefore evaluated
on the restricted datasets described in \Cref{sec:setup:data}. We rerun
NFPA on compatible inputs. For Boundary Leakage, we compare the authors'
released attacked images with an independently generated clean SD2.1
reference set.
Both attacks remain readily detectable under their respective evaluation constructions
(\Cref{tab:unified-detection}).

\paragraph{Stochastic erosion (Watermarks in the Sand).}
WiTS~\cite{zhang2024watermarks} is the stochastic attack in our
evaluation: it repeatedly inpaints small random patches with Stable
Diffusion inpainting~\cite{ldm_cvpr22} and keeps a candidate only when its HPSv2
score~\cite{wu2023human} is at least that of the original watermarked
image.
In Zhang et al.'s construction~\cite{zhang2024watermarks}, watermark
erosion arises as a byproduct of this oracle-guided drift rather than
of direct optimization.

Because each image requires 50 sequential inpainting forward passes,
the attack is computationally expensive, so we evaluate it on a subset
of the base dataset of approximately $29{,}000$ images
(\Cref{sec:setup:data}). As \Cref{tab:unified-detection} shows, WiTS remains strongly detectable,
with AUROC $0.9979$ and TPR@0.1\%\,FPR of $96.23\%$.

\paragraph{Qualitative comparison.}
\Cref{fig:cross-attack-comparison} in
\Cref{app:visual-comparison} shows one illustrative source--output pair
for each attack. These particular outputs remain visually plausible and
closely preserve their source content to a human observer. Quantitative
fidelity and pass counts across the broader evaluation are reported in
\Cref{tab:unified-detection,tab:joint-outcome}.

\paragraph{Synthesis.}
The six attacks span four families, and in each known-remover evaluation the detector distinguishes the evaluated output set from its clean reference. The pooled representation also differs across the evaluated output sets (\Cref{fig:tsne-pooled}). Because those sets do not all share the same source distribution, this structure cannot be attributed solely to the removal mechanism. Within these evaluated distributions, the removers do not return images to a clean forensic state. They often trade an explicit watermark for an \emph{implicit watermark}: a learnable signal associated with the evaluated removal pipeline.

\subsection{Representation controls and pipeline integrity}
\label{sec:results:validations}

A detector could appear to succeed by exploiting file-format metadata,
compression side channels, or other non-pixel cues rather than genuine
processing traces. We therefore ran the same validation checks across
all six attack detectors under the shared protocol. We report the
numbers in detail for UnMarker~\cite{unmarker_sp25}, the main case
study used throughout the rest of the section, and summarize the
corresponding cross-attack results in
\Cref{app:cross-attack-controls}. The central control is BMP re-encoding, which removes container and
file-encoding cues and fixes file size without undoing artifacts already
present in the decoded pixels. The remaining checks ask whether the
conclusion survives re-export, metadata stripping, split audits, and
more aggressive image simplifications.

\paragraph{BMP re-encoding and file-size baseline.}
We re-encode every decoded image as an uncompressed bitmap (BMP), which
removes container, metadata, and file-encoding cues and produces a
constant file size at fixed resolution. For UnMarker~\cite{unmarker_sp25}, the detector remains at AUROC~0.9998 and TPR@0.1\%\,FPR~99.76\%, unchanged to four decimal places (\Cref{tab:cross-attack-controls} in \Cref{app:cross-attack-controls}). A classifier using \emph{only file size}
achieves AUC~$= 0.765$ in the native setting (clean images average
366\,KB; attacked images average 485\,KB), but collapses to chance
under BMP (AUC~$= 0.500$). This is the decisive control in the
subsection: once storage cues are removed, the pixel-based detector is
unchanged while the size-only baseline disappears.

\paragraph{Supporting audits.}
Re-encoding all images with identical PNG compression settings after stripping metadata produces metrics identical to native evaluation (AUROC~$= 0.9998$, TPR@0.1\%\,FPR~$= 0.9976$, identical confusion matrix), ruling out encoder-specific quirks as a source of
separability. The match is exact: no test
image crosses the decision threshold under re-export. We also
inventoried PNG container fields, including ancillary chunks, text metadata, and color profiles, and found no class-specific leakage patterns. Split integrity is intact: no file path or image identifier appears in more than one split in any dataset manifest, and zero exact cross-split duplicates were found across all attack datasets, excluding NFPA (\Cref{app:cross-attack-controls}). See \Cref{app:cross-attack-controls} for the full per-attack breakdown and a walkthrough of what each control rules out.

These controls show that the observed separability arises from the
pixel content produced by the remover itself.

\paragraph{Benign generative controls.}
The matched-light controls of \Cref{sec:setup:data} test only whether
light benign processing is sufficient to trigger the detectors; they do
not isolate removal-specific traces. As an auxiliary shortcut check, at a common score threshold of $0.5$, the recovered controls fire on
$1/244$ Boundary Leakage twins, $36/3{,}036$ CtrlRegen+, $5/416$ WiTS,
$49/1{,}231$ NFPA, and $82/2{,}975$ WMA, a range of $0.41\%$ to $3.98\%$.
UnMarker has no generative stage and therefore no twin.

\subsection{Exact empirical three-gate outcome}
\label{sec:results:joint}

Marginal rates do not show whether the same output passes all three security gates. We therefore evaluate five attack-specific samples, each containing
$150$ genuinely watermarked inputs. Removal, acceptable LPIPS, and detector evasion are measured on the same outputs using
validation-frozen $1\%$-FPR detector thresholds. No threshold is selected or adjusted using these $750$ outputs. All joint-evaluation outputs and their progenitors are disjoint from detector training, validation, and threshold calibration.

\begin{table}[t]
  \centering
  \footnotesize
  \caption{\textbf{Exact empirical three-gate outcome.} Within each row, 
  all columns are computed on the same $150$ outputs. ``Removed'' means verifier rejection,
  ``LPIPS'' means $\mathrm{LPIPS}\leq0.35$, and ``Evaded'' means falling
  below the detector threshold. The three-way count is their same-image
  conjunction, not a product of marginal rates. For $0/150$, the
  one-sided 95\% Clopper--Pearson upper bound is $1.98\%$; the WiTS
  $1/150$ result has a two-sided 95\% interval of $[0.02,3.66]\%$.}
  \label{tab:joint-outcome}
  \setlength{\tabcolsep}{2.5pt}
  \begin{tabular*}{\columnwidth}{@{\extracolsep{\fill}}lrrrr@{}}
    \toprule
    \textbf{Attack}
      & \textbf{Removed}
      & \textbf{LPIPS}
      & \textbf{Evaded}
      & \textbf{3-way} \\
    \midrule
    UnMarker   & $150/150$ & $150/150$ & $0/150$ & \textbf{$0/150$} \\
    CtrlRegen+ &  $96/150$ & $150/150$ & $0/150$ & \textbf{$0/150$} \\
    WMA        &   $1/150$ & $150/150$ & $0/150$ & \textbf{$0/150$} \\
    NFPA       & $150/150$ &  $6/150$ & $0/150$ & \textbf{$0/150$} \\
    WiTS       & $150/150$ &  $64/150$ & $1/150$ & \textbf{$1/150$} \\
    \bottomrule
  \end{tabular*}
\end{table}

Across the five attack-specific cohorts, only a single WiTS output passes all three gates, giving $1/750$ overall. This aggregate is descriptive because the rows represent different attack distributions, and the single success does not establish that WiTS matches the clean distribution. The table defines a forward-looking evaluation target rather than a uniform comparison: WMA was not designed around the deployed PRC verifier, while NFPA's inversion-based transformation rarely passes the LPIPS gate.

\subsection{Robustness to post-processing}
\label{sec:results:robustness}

Post-processing is already part of our main evaluation. We apply the same ten lightweight operations to clean and attacked images, so the operation itself does not determine the class label. Here we report UnMarker~\cite{unmarker_sp25} performance by operator and examine JPEG recompression at fixed quality levels. Across the full post-processed set, the detector reaches AUROC $0.9964$.

The ten operations cover geometric transformations, color changes,
smoothing, and compression. \Cref{tab:tpr-by-operator} reports detection
performance separately for each type of operation.

\begin{table}[t]
\centering
\ifconf\small\else\footnotesize\fi
\caption{\textbf{Detection sensitivity by post-processing operator.}
TPR on the operation-specific subsets of the UnMarker test set. The six geometric and
color-space operations share one row because each reaches 1.000 TPR
at both operating points.}
\label{tab:tpr-by-operator}
\setlength{\tabcolsep}{4pt}
\begin{tabular*}{\columnwidth}{@{\extracolsep{\fill}}lcc@{}}
\toprule
\textbf{Operator}
& \makecell{\textbf{TPR@}\\\textbf{1\%\,FPR}}
& \makecell{\textbf{TPR@}\\\textbf{0.1\%\,FPR}} \\
\midrule
Bilateral filter        & 0.96 & 0.15 \\
JPEG recompression      & 0.70 & 0.41 \\
Non-local means         & 0.83 & 0.61 \\
Gaussian blur           & 0.88 & 0.72 \\
\midrule
Six geometric and color-space operations
                        & 1.00 & 1.00 \\
\bottomrule
\end{tabular*}
\end{table}

All six geometric and color-space operations preserve complete detection
at both operating points. Smoothing and recompression weaken the detector:
at $1\%$ FPR, JPEG produces the lowest TPR ($0.70$), whereas at $0.1\%$
FPR bilateral filtering is hardest ($0.15$). Robustness therefore depends
on the operating point. All four weakening operations alter spatial detail. Their relationship
to the residual-spectrum profile in
\Cref{sec:signature:unmarker} is descriptive; the present analysis does
not identify why detection degrades. Because the evaluated output sets have different residual profiles, the post-processing operations that most weaken their detectors may also
differ.

\paragraph{JPEG recompression.}

JPEG is already one of the ten evaluated operations. We examine it
separately because recompression is common when images are stored or
shared, and its strength can be varied systematically. We evaluate
qualities 90, 85, and 75; \Cref{tab:jpeg-sweep} reports the resulting
detection performance.

Using the earlier checkpoint described in \Cref{app:jpeg}, the held-out
test ROCs give TPRs of $96.6\%$ at Q90 and $93.3\%$ at Q85, falling to
$74.5\%$ at Q75, when read at $1\%$ FPR. The stricter 
$0.1\%$-FPR point is more fragile: TPR falls from $98.3\%$ natively 
to $91.9\%$ at Q90 and $27.8\%$ at Q85, remaining near $28\%$ at Q75. 
Importantly, AUROC at Q85 remains $0.9941$. Recompression therefore 
weakens extreme-tail sensitivity without eliminating the broader 
separation between clean and UnMarker-processed images.

The evaluated detector is trained jointly across a mixture
of ten post-processing operations, rather than optimized for a fixed recompression channel. Because 
matched JPEG-recompressed clean and attacked examples are inexpensive to 
generate, a deployment could train or fine-tune specifically at its 
expected quality levels; the high AUROC indicates that substantial 
separability remains available to such a detector. We have not evaluated 
this quality-specific intervention, so we present it as a plausible 
mitigation rather than a demonstrated fix. Detailed sweep results appear in \Cref{app:jpeg}.

\section{Forensic Signature Analysis}
\label{sec:signature}

All six evaluated removal output sets are distinguishable from their
clean references. We next compare their residual spectra and learned representations.

For five attacks, we compute pixel-domain residuals between source
images and their attack-processed outputs; for Boundary
Leakage~\cite{boundary_leakage}, we use the released
watermarked--attacked pairs. We compute their azimuthally averaged power
spectral density (PSD) and compare it with differences between
independent clean images from the same source distribution. We use 
$N=5{,}000$ comparisons per attack, except WiTS, where $N=2{,}000$, 
and Boundary Leakage, where $N=4{,}997$. The resulting log-ratio 
shows where an attack residual has more or less power than ordinary 
inter-image variation. It does not show that the attacked image itself 
gains or loses energy at those frequencies. The residual construction and two-dimensional map appear in
\Cref{app:spectral-methods}.

\subsection{Case study: UnMarker}
\label{sec:signature:unmarker}

UnMarker~\cite{unmarker_sp25} provides a useful case study because its
residual-to-reference profile differs from those of the other attacks
(\Cref{fig:radial-profiles}). Its curve lies above the inter-image
reference at lower frequencies and below it at higher frequencies.
Thus, the transformation residual has greater power than ordinary
inter-image differences in the former region and lower power in the
latter. 

This profile suggests a possible connection to the post-processing
results in \Cref{sec:results:robustness}. Bilateral filtering removes
fine texture while preserving edges~\cite{tomasi1998bilateral}, and JPEG quantization attenuates
mid-to-high-frequency detail~\cite{wallace1991jpeg}. UnMarker's residual-to-reference profile
also differs from the reference at higher frequencies. This overlap is
consistent with, but does not establish, the possibility that these
operations weaken detector-relevant high-frequency cues. At the
$0.1\%$ FPR operating point, this weakening is especially pronounced
for JPEG, with TPR falling sharply by Q85
(\Cref{tab:jpeg-sweep}).

\subsection{Spectral signatures across attack families}
\label{sec:signature:cross}

\begin{figure}[t]
  \centering
  \includegraphics[width=\linewidth]{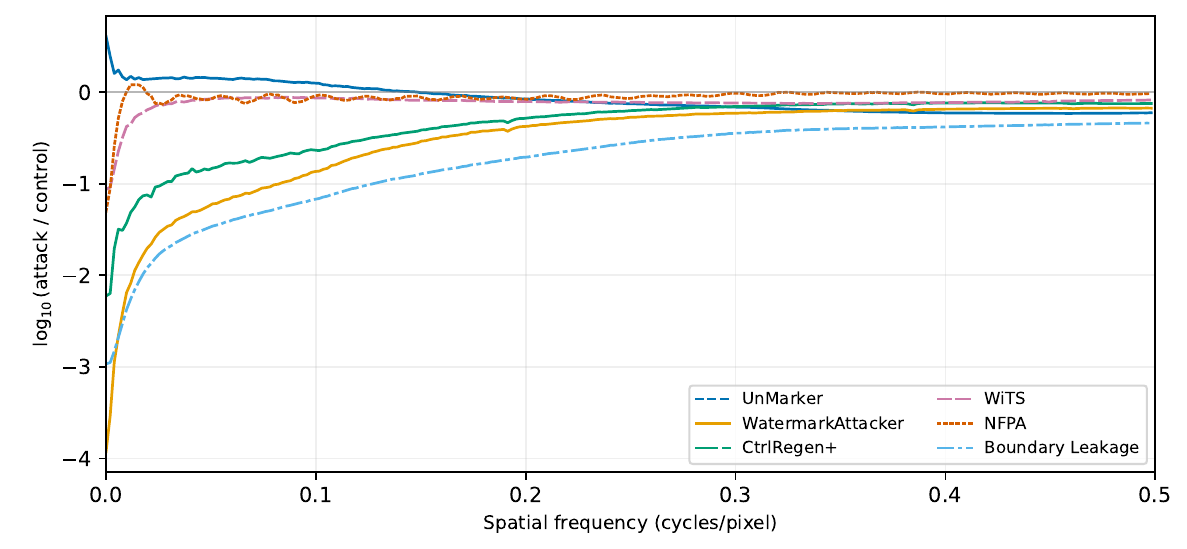}
  \caption{%
  \textbf{Cross-attack residual spectra.}
  For each attack, the figure shows the azimuthally averaged base-10
  log-ratio of its mean residual PSD to a source-matched inter-image
  reference. Positive and negative values indicate more or less residual
  power than the reference. Boundary Leakage starts from watermarked
  images, so its magnitude is not directly comparable to the others.}
  
  \label{fig:radial-profiles}
\end{figure}

\begin{figure}[t]
  \centering
  \includegraphics[width=0.8\linewidth]{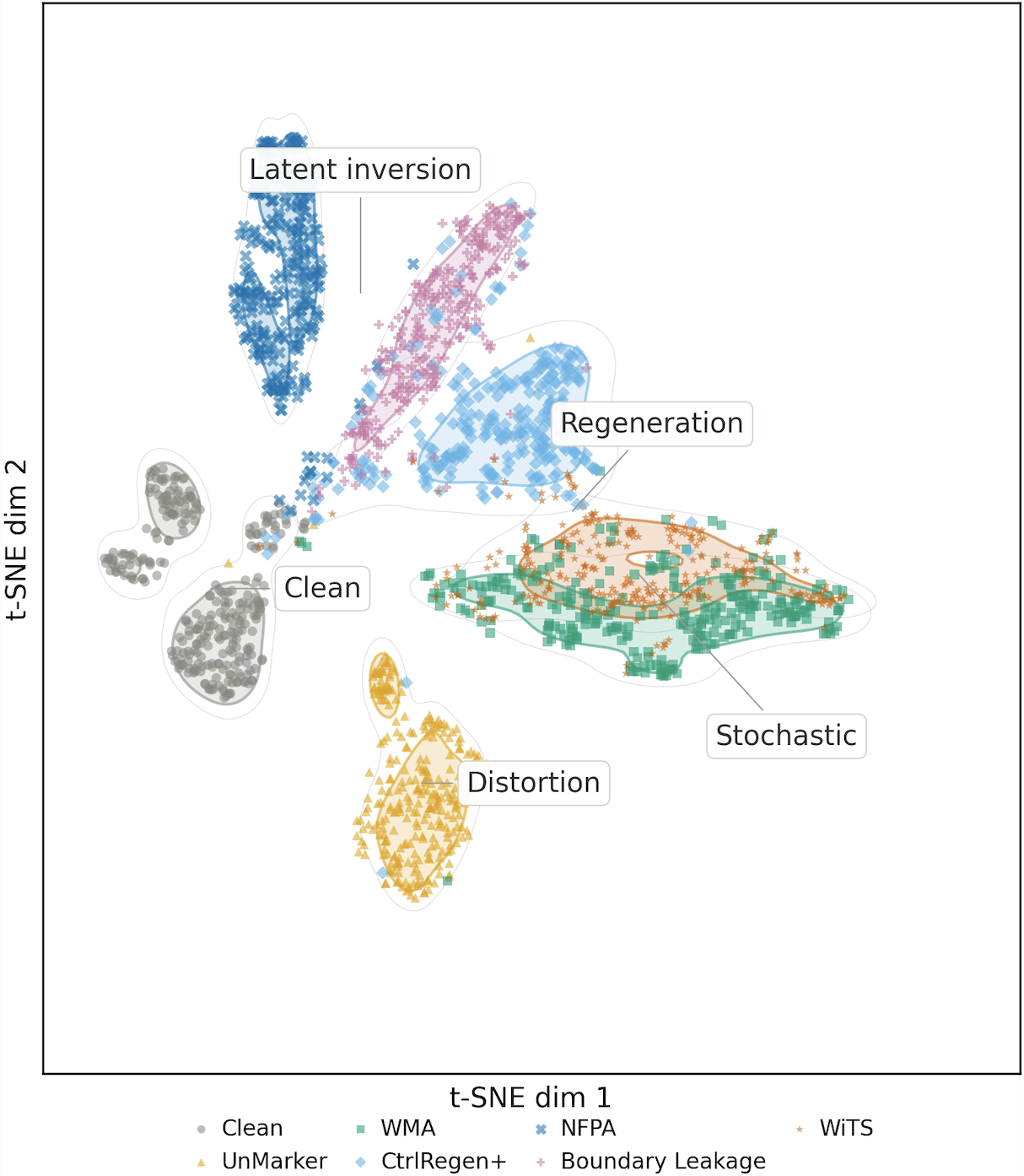}
  \caption{%
\textbf{Feature structure across the evaluated output sets.}
t-SNE~\cite{vandermaaten2008tsne} projection of pooled-detector
features from $2{,}100$ held-out images: 300 clean images and 300
outputs from each attack. The detector is trained only to distinguish
clean from processed images, yet its representation separates the
seven evaluated sets. Because some attacks use different source
populations, this visualization describes output-set structure rather
than isolating attack mechanism. The projection uses an earlier pooled checkpoint than the one reported in
\Cref{tab:unified-detection}; the checkpoints use the same training
configuration but different progenitor splits, and the figure is used only
as a qualitative visualization. Projection details appear in
\Cref{app:pooled-detector}.}
  \label{fig:tsne-pooled}
\end{figure}

UnMarker stands out among the six attacks because its radial
residual-to-reference curve crosses zero
(\Cref{fig:radial-profiles}), whereas the other five curves lie mostly below it. A zero log-ratio denotes equal mean power in the attack residual and its source-matched reference. Values above or below zero indicate that the attack residual has more or less power than that reference, respectively. They do not mean that the attacked image gains or loses energy at those
frequencies.

WatermarkAttacker and WiTS have similarly shaped profiles despite using
different removal procedures, while CtrlRegen+ and Boundary Leakage
produce different patterns. This shows that the curves can serve as
descriptive fingerprints of the evaluated outputs, but cannot identify
the mechanism that produced them. 

Radial averaging hides directional information. The two-dimensional
maps in \Cref{fig:2d-logratio} reveal axis-aligned patterns in several
attacks and a periodic grid for NFPA that is much less visible in its
radial profile. The radial and two-dimensional views therefore capture
complementary structure in the residuals. NFPA is nevertheless detected at $93.73\%$ TPR at the $0.1\%$ FPR
operating point (\Cref{tab:unified-detection}). Its comparatively
minute radial profile therefore shows that strong detectability
need not coincide with a large radial spectral deviation.

The pooled detector's features also separate the seven evaluated output
sets, despite being trained only for binary clean-versus-processed
classification (\Cref{fig:tsne-pooled}). The t-SNE projection is used only for visualization; it does not affect detector training, threshold selection, or any reported operating point. To quantify this structure,
we freeze the features and train a linear classifier to distinguish
clean images from the six attack output sets. The classifier reaches
$0.971$ balanced accuracy, compared with $0.600$ using generic
ImageNet-1K features and $0.143$ for random guessing. Thus, the
removal-trained representation contains substantially more information
about the evaluated output sets than generic image features
(\Cref{app:pooled-detector}).

Together, these analyses reveal substantial structure and differences  among the evaluated output sets in both residual spectra and learned  representations. We also evaluate a higher-order whitening attempt  designed to align spectral, residual, and bispectral statistics with a  clean reference. Its corrected optimized statistic instead moves from  $0.842$ to $1.093$, farther from the clean target, so the intended match  is not achieved. These outputs remain separable under the held-out four-detector battery, which achieves an AUROC of $1.000$  (\Cref{tab:stealth-attempts} in \Cref{app:stealth-attempts}).

\section{Toward Forensic Stealth}
\label{sec:discussion}
\label{sec:stealth}

A successful remover should leave an image that could have arisen from ordinary
processing of the same source. Silencing the watermark verifier is part of this
task, but it does not restore deniability when an investigator can still
recognize the removal process. Our experiments show how far present attacks are
from this goal. Among 750 outputs, only one removed the watermark, remained
within the fidelity budget, and evaded the attack-specific forensic detector.
Although this result does not rule out a better remover, it asks us to treat
stealth as part of removal itself and to understand what such a remover would
have to reproduce.

\subsection{When removal can hide in ordinary processing}
\label{sec:discussion_implications}

A remover has a natural place to hide when the same image operation also occurs
during ordinary processing. Recent work on PRC-style watermarking gives a
concrete example~\cite{francati2025codinglimits}. Cropping followed by resizing
erases the watermark while leaving the image visually close to the original.
The same work considers a model that first generates an artificial border. The
attacker crops away the border and resizes the image, preserving the central
content while the global resampling defeats watermark recovery.

Crop and resize can therefore remove the watermark without obvious damage to
the central image. A forensic detector might nevertheless recognize that the
operation occurred. Whether this trace reveals an attack depends on how images
are normally released. If legitimate images are commonly cropped and resized,
the operation alone cannot reveal why it was applied, and the remover can blend
with benign processing. The harder case arises when no ordinary operation both
removes the watermark and preserves the image.

\subsection{What a stealthy remover must reproduce}
\label{sec:discussion_stealth}

A clean source has more than one legitimate digital form. The same photograph
may be resized, compressed, adjusted in color, or slightly cropped before
release. These images differ at the pixel level, yet each remains an ordinary
version of the same photograph. An unrelated natural image may also look
benign, but returning it would abandon the content that the attacker is meant
to preserve.

Fix a clean source image $x$, and let $Q_x$ be the distribution of its
ordinary, non-attacked releases. Repeated applications of a randomized remover
$K$ to the watermarked image $x_w$ induce another distribution, denoted by
$\mathsf{Rem}_x^K$. Forensic stealth asks the remover to reproduce the kinds of
variation found in $Q_x$, in roughly the same proportions. Let $\pi$ be the source distribution shared by clean and removed releases. Total variation measures the largest difference assigned to the same image set, and we use the stronger, information-theoretic
sufficient condition
\[
  \mathbb{E}_{x\sim\pi}
  \left[\Delta(\mathsf{Rem}_x^K,Q_x)\right]\leq\varepsilon.
\]
When clean and removed releases use the same source distribution $\pi$,
averaging $Q_x$ and $\mathsf{Rem}_x^K$ over $x\sim\pi$ yields
$\mathbb{P}_{\mathrm{clean}}$ and $\mathbb{P}_{\mathrm{rem}}^K$,
respectively. By convexity, the average conditional criterion also bounds
their total-variation distance and therefore the global distinguishing
advantage in the threat model. Under this condition, the acceptance
probabilities of any image-only test differ by at most $\varepsilon$
between removed and ordinary releases. If $\varepsilon$ is negligible
in the security parameter, this implies the
computational-indistinguishability target. Because
$Q_x$ is tied to the source, a clean image from somewhere else is not an
acceptable substitute.

A search could compare candidate images through the engineering objective
\[
 S_x(y)=\lambda_{\mathrm{fid}}F_x(y)
       -\lambda_{\mathrm{wm}}W(y)
       +\lambda_{\mathrm{ben}}B_x(y).
\]
Here $F_x$ rewards preservation of the source, $W$ measures the watermark
evidence that remains, and $B_x$ rewards resemblance to ordinary releases of
$x$. As a pointwise search heuristic, $S_x$ may concentrate outputs on a
narrow set and does not replace distributional evaluation against $Q_x$. Whatever value $S_x$ takes, the final experiment applies every gate to the
same image, which still fails if the watermark remains detectable or the
removal trace remains obvious. The baseline attacker cannot compute $W$ by
querying the secret verifier. In this expression, $W$ records the desired
absence of watermark evidence. Verifier queries or feedback from a learned
watermark proxy would give the attacker more power than our baseline model
allows.

An ordinary transformation can do both jobs when it weakens the watermark and
also belongs to the release process represented by $Q_x$. If no such
transformation is available, the remover must introduce plausible variation of
its own, perhaps through a conditional editor or a stochastic search among
nearby images. WiTS performs such a search by making a sequence of small random
edits under a quality oracle~\cite{zhang2024watermarks}, but the random walk has
no model of the clean distribution toward which its candidates should move.
We therefore added forensic scores and asked whether they could steer the walk
toward outputs that are less recognizable as attacks.

In a three-way optimization search, we use WiTS to determine which candidate images are ever produced, while
the forensic surrogates help decide which of those candidates is worth keeping.
The distinction matters when the search fails, because no ranking rule can
recover an image that the random walk never visits, and a good proposal may
still be lost when the surrogates rank it poorly.

At each of 200 steps, WiTS produces three nearby proposals. The original WiTS
quality condition keeps a proposal only when its HPSv2 score is at least that
of the watermarked input~\cite{wu2023human}. Among the eligible proposals,
tri-search combines HPSv2 with a forensic benignity score. Two attacker-trained
surrogates supply this score, namely a
ViT-S/16~\cite{dosovitskiy2021image} and a classifier using radial-spectrum
features. Neither surrogate is conditioned on the source, so they judge whether
a candidate resembles attack outputs in general without explicitly modeling
ordinary variation for that particular source. The search saves checkpoints
and finally selects the one with the best surrogate benignity score among those
that satisfy $\mathrm{LPIPS}\leq0.35$.

The search has no direct signal for watermark removal, so the accumulated WiTS
edits must weaken the watermark on their own. It never queries the secret
verifier or the final cross-architecture ResNet-50 detector. These tests are
applied only after tri-search has selected its output, while candidate
selection relies on the two attacker-controlled surrogates.

All 30 selected outputs stayed within the LPIPS fidelity budget. The verifier
rejected 26, but the unqueried ResNet-50 caught every one of them at its $1\%$
FPR threshold. The remaining four evaded ResNet-50, but the verifier still
detected their watermarks, leaving no image that passed all three gates.

WiTS may never have generated a suitable candidate, leaving the ranking rule
nothing to recover. Even if such a candidate was present, the source-agnostic
surrogates may have ranked it poorly, or their scores may not have transferred
to the unqueried ResNet-50. A larger search budget might also have produced a
different result. Because the verifier and final detector were applied only to
the selected outputs, we cannot tell which explanation is right. We therefore
separate the failure of this search from the more basic question of whether the
three requirements can be met at all. We consider this possibility next through an idealized construction of a stealthy remover.

\subsection{Why the three goals can coexist}
\label{sec:stealth:boundary}

Suppose that a source has several legitimate clean versions, all preserving
the content that matters for the application. A remover might replace the
watermark-bearing detail with variation from this family while staying inside
the distortion budget. Its outputs would then follow the ordinary clean
variation of the source instead of forming a separate removal distribution
that a detector could learn.

Represent an image as $(s,c)$, where $s$ contains the source information
that must be preserved at the application's required fidelity level and
$c$ contains only detail that may be replaced within the distortion
budget. Let $P$ be the clean
distribution of $s$, and let $R_s$ be the clean distribution of detail for
that content. In this simplified population model, $R_s$ describes ordinary
variation after the source content has been fixed. Clean images follow the
distribution
\[
  \nu(ds,dc)=P(ds)R_s(dc).
\]
For every watermark key $k\in\mathcal K$, let $M_{k,s}$ describe the
watermarked detail conditioned on $s$. We assume that watermarking leaves the
content marginal equal to $P$, and hence
\[
  \mu_k(ds,dc)=P(ds)M_{k,s}(dc).
\]

For a photograph, $s$ may contain the objects and scene structure that give the
source its identity, while $c$ collects the smaller choices through which the
same source can appear in different digital forms. The appropriate division
depends on the application and its distortion measure, since these determine
what the remover is allowed to change.

We consider one efficient public randomized rule $T$ that is independent of
the watermark key and works for every relevant key. Requiring one common rule
is important because a keyless attacker cannot choose a different procedure
after learning the secret key. Given the preserved content $s$ and the current
detail $c$, the rule produces a replacement according to $T_s(c,\cdot)$. Every
replacement that it may produce must stay within the distortion budget with
probability one. After the rule is applied to watermarked details generated
under any relevant key, the resulting details must follow exactly the clean
conditional law $R_s$. For every relevant key $k$ and state $s$, every detail $c$ in the
support of $M_{k,s}$, and every measurable set $A$ of details, these
requirements are
\[
\begin{aligned}
T_s\!\left(c,\{c' \mid d((s,c),(s,c'))\leq\tau\}\right)&=1,\\
\int M_{k,s}(dc)\,T_s(c,A)&=R_s(A).
\end{aligned}
\]
Allowing $T$ to depend on the original detail $c$ lets it choose a nearby clean
replacement instead of an independent version that may exceed the distortion
budget. The following proposition shows that these properties are sufficient
for exact population stealth.

\begin{samepage}
\begin{proposition}[Stealth by conditional resampling]
\label{prop:clean-fiber}
Assume that the content-detail representation and the replacement rule above
can be implemented efficiently. Let $K$ preserve $s$ and draw $C'$ from
$T_s(c,\cdot)$. Then, for every $k\in\mathcal K$, the output remains within the
distortion budget with probability one and
\[
  K\mu_k=\nu.
\]
Writing
\[
  \beta_k=
  \Pr_{Y\leftarrow\nu}
  \left[\mathsf{Verify}(Y;k)=1\right]
\]
for the verifier's false-positive probability on clean images, the verifier
rejects the remover's output with probability $1-\beta_k$.
\end{proposition}
\end{samepage}

\begin{proof}
Fix a watermark key $k$. The content of a watermarked input is distributed
according to $P$, and the remover leaves this content unchanged. For a fixed
value of $s$, the second property of $T$ gives
\[
  \int M_{k,s}(dc)\,T_s(c,dc')=R_s(dc').
\]
The output distribution is consequently
\[
  (K\mu_k)(ds,dc')
  =P(ds)R_s(dc')
  =\nu(ds,dc').
\]
The first property of $T$ gives the distortion guarantee. Since the output has
the clean distribution, the verifier accepts with probability $\beta_k$ and
rejects with probability $1-\beta_k$.
\end{proof}

Because $K\mu_k=\nu$, an image-only test sees exactly the same distribution on
removed outputs and clean data, even if it is computationally unbounded and
knows the watermark key. Some outputs may still be flagged, but over random population draws this happens at the test's clean false-positive rate. The keyed verifier
likewise accepts with probability $\beta_k$.

A simple idealized model shows that these conditions can hold. Let the detail decompose as $c=(h,z)$, where $h$ is benign detail and $z\in\{0,1\}^r$ is a replaceable component that may carry the watermark. Let $R_s=H_s\otimes U_r$, where $H_s$ is the clean law of $h$ and $U_r$ is uniform on $\{0,1\}^r$. Suppose every $M_{k,s}$ has $h$-marginal $H_s$, while the conditional law of $z$ may depend on $h$ and the key.
Assume also that
\[
d\bigl((s,(h,z)),(s,(h,z'))\bigr)\leq\tau
\]
for every $h,z,z'$. Define $T_s$ to preserve $h$ and replace $z$ with an independent draw from $U_r$. This rule is efficient, public, and key-independent. It respects the distortion budget, preserves the $H_s$ marginal, and makes the new $z$ independent and uniform. Its output law is
therefore $H_s\otimes U_r=R_s$ for every key. Thus, the proposition's assumptions are jointly satisfiable in this model.

We do not know how to construct $T$ for Stable Diffusion or for the watermark
systems we evaluate. Such a rule would need to determine which content must
remain fixed and learn the clean variation appropriate to that content, so
that it can sample within a tight distortion budget without knowing the key.
An efficient, key-independent sampler with these properties would meet all
three requirements, and constructing one remains open.

\section{Related Work}
\label{sec:related}

This paper draws on three adjacent literatures: watermarking and provenance for AI-generated images, attacks that remove these signals, and forensic analysis of image-processing pipelines. Our focus is the output after watermark removal, where these threads intersect.

\paragraph{Watermarking, provenance, and formal guarantees.}
Research on watermarking and provenance asks how provenance signals are embedded, how well they survive transformations, and what guarantees these systems can provide. Classical post-generation methods include DWT--DCT--SVD
watermarking~\cite{dwtdctsvd} and learned encoders such as
HiDDeN~\cite{hiddenwm_eccv18}. Later methods embed the signal during generation, as in Tree-Ring~\cite{treering_neurips23}, Gaussian
Shading~\cite{gaussianshading_cvpr24}, and Stable
Signature~\cite{stablesig_iccv23}. More recent systems broaden the
design space further, including ZoDiac~\cite{zodiac_neurips24},
ROBIN~\cite{robin_neurips24}, RAW~\cite{raw_neurips24},
SEAL~\cite{seal_iccv25}, InvisMark~\cite{invismark_wacv25}, and
TrustMark~\cite{trustmark_iccv25}. A complementary line studies
robustness under stronger transformations, including
instruction-driven editing in Robust-Wide~\cite{robustwide_eccv24} and
broader editing and regeneration in VINE and
W-Bench~\cite{vine_iclr25}.

Aaronson~\cite{Aaronson22} gave an early account of watermarking model
outputs using pseudorandom signals embedded during sampling. Later
formal work develops this view through pseudorandom coding,
computational undetectability, and coding
limits~\cite{C:ChrGun24,gunn2025undetectable,francati2025codinglimits}.
Recent cryptanalysis shows that the practical LDPC-PRC instantiation can
be distinguished at a cost of approximately $2^{22}$ operations, limiting
its concrete security under feasible parameters
\cite{wang_ldpc_prc_usenix26}. Related information-theoretic work studies
exact distribution matching through perfectly secure steganography and
coupling-based constructions
\cite{wang2008perfectly,schroederdewitt2023perfectly}. At the deployment
level, SynthID-Image and C2PA represent watermark-based and
signed-metadata-based approaches to
provenance~\cite{synthid_image_arxiv25,c2pa_spec}. Our question begins
after removal: what forensic evidence does the removal process itself
leave behind?

\paragraph{Watermark removal attacks and their limits.}
The six attacks we evaluate span
four families: distortion-based optimization~\cite{unmarker_sp25},
diffusion-based regeneration~\cite{watermarkattacker_neurips24,
ctrlregen_iclr25}, latent-space inversion and
perturbation~\cite{nfpa_attack,boundary_leakage}, and stochastic
erosion~\cite{zhang2024watermarks}. Other attacks include detector optimization in WAVES~\cite{waves_arxiv24}, single-image Deep Image Prior
removal~\cite{dip_watermark_tmlr25}, the winning NeurIPS black- and
beige-box challenge pipeline~\cite{neurips_removal_challenge25},
black-box semantic forgery and removal~\cite{blackbox_forgery_cvpr25},
MarkSweep's frequency-aware denoising~\cite{marksweep_arxiv26}, and
TrustMark's ReMark remover~\cite{trustmark_iccv25}.

Zhang et al.~\cite{zhang2024watermarks} prove that strong
watermarking is impossible under stated assumptions on quality and
perturbation oracles; their constructive attack, WiTS, is one of the
six removers we evaluate. Their result addresses watermark removal
while preserving output quality, but not whether the resulting images
resemble ordinary clean releases.

Independent and concurrent work by Evennou and Kijak studies the same
phenomenon under the name Watermark Removal Detection and concludes that removal can leave a recognizable forensic trace~\cite{forensic_cost_arxiv26}. Their study provides further evidence across dedicated attacks and image-editing pipelines. Our work gives a complementary security treatment for keyed PRC watermarking: we formalize stealth as computational indistinguishability from a specified clean-release distribution and report the exact same-image conjunction of verifier
evasion, source fidelity, and detector evasion. We additionally evaluate benign-processing controls and adaptive detector evasion followed by fresh-arbiter retraining, and investigate whether a remover can be designed to meet this stronger objective. 

Recent work by Cao et al.~\cite{marknull_usenix26}, presented at USENIX Security 2026, introduces MarkNull, a model-agnostic latent-space removal attack, and MarkNull-A,
a learned variant that performs the transformation in a single forward
pass. The paper also evaluates a separate reconstruction-based detector:
it inverts and regenerates a query image, then treats unusually large
perceptual reconstruction error as evidence of attack. This experiment
distinguishes untouched watermarked images from outputs of MarkNull and
MarkNull-A. The attacks themselves are optimized for watermark evasion
and similarity to the input; forensic detectability is not part of their
attack-success criterion. Our evaluation
instead requires watermark evasion, quality preservation, and forensic
evasion to hold jointly.

\paragraph{Image forensics and learned pipeline fingerprints.}
Image forensics studies the statistical traces left by image-processing
and generation pipelines. Bayar and Stamm~\cite{bayar_stamm_ihmms16}
show that detectors can learn traces of several image manipulations
directly from data, without a hand-designed feature for each operation.
Marra et al.~\cite{marra2019gans} identify model-specific fingerprints
in GAN outputs, while Wang et al.~\cite{wang2020cnn} show that a detector
trained on one generator can transfer to previously unseen generators.
Frank et al.~\cite{frank2020leveraging} and Durall et
al.~\cite{durall2020watch} trace many of these signals to spectral
artifacts introduced by upsampling. Corvi et
al.~\cite{corvi_cvprw23} extend this spectral analysis to diffusion
models and find that they also leave distinct Fourier, radial, and
angular patterns. Ojha et al.~\cite{ojha2023universal} further improve
transfer across generator families using pretrained features.

Together, these works provide the forensic foundation for our study.
These studies mainly compare generated images with real ones, or
manipulated images with untouched ones. We instead compare
removal-processed images with ordinary clean releases and ask whether
the resulting signatures reveal the mechanism used for removal.

\section{Conclusion}
\label{sec:conclusion}

Watermark removal is usually judged by verifier evasion and image
quality. Our results show that these criteria do not restore deniability:
across the six removers we study, outputs remain highly distinguishable
from clean content, and across the five executable attacks, only $1/750$
outputs in our exact same-image evaluation satisfy removal, source
fidelity, and detector evasion simultaneously. This does not prove forensic stealth impossible, but shows that evading a watermark verifier or fixed forensic detector is insufficient. We therefore argue that removal should be evaluated against the clean-release distribution itself, and provide both a practical objective for pursuing this goal and an idealized model satisfying sufficient conditions for exact population stealth. Ultimately, successful removal requires not only erasing the watermark, but making the image indistinguishable from an ordinary clean release.

\ifconf\else
  \pdfbookmark[1]{Ethical Considerations}{ethical-considerations}
\fi
\section*{Ethical Considerations}

We study detection of watermark-removal outputs and attempts to evade
those detectors. Five attacks have public implementations; Boundary
Leakage is evaluated from author-released outputs. The same-image
tri-search yields no three-gate successes in 30 trials. Other adaptive
ablations measure detector evasion and retraining, not complete
three-gate success. No human subjects are involved, and all data are
public or generated from public prompts and models. Detector-evasion
methods may aid provenance laundering or disinformation, harming
creators, platforms, watermark providers, and users who rely on
provenance signals. We report these methods because reproducible
adaptive evaluation is necessary to assess the security claim, but
withhold secret watermark keys to limit exact reproduction of
key-dependent attacks. This reduces but does not eliminate misuse.

Because benign processing may resemble removal, false accusation is a
central deployment risk. We evaluate family-specific benign controls
(\Cref{sec:results:validations}); detector scores should be treated as
screening signals, not as proof of malicious removal or an image's
provenance.

\ifconf\else
  \pdfbookmark[1]{Open Science}{open-science}
\fi
\section*{Open Science}

\ifconf
The anonymous artifact is available at
\url{https://anonymous.4open.science/r/removal-is-not-stealth-1000/}.
\else
The artifact is available at
\url{https://anonymous.4open.science/r/removal-is-not-stealth-1000/}.
\fi
It contains manifests, splits, checkpoints, per-image scores, evaluation
scripts, and parameters for five public attacks, plus Boundary Leakage output
provenance.

We withhold watermark keys but release verifier outcomes. Generated datasets
include prompts, model identifiers, sampler settings, and the pinned Stable
Diffusion~2.1 mirror revision. Scores support metric recomputation and
checkpoints support inference on reported splits. Because no global seed was
fixed, exact checkpoint reproduction is not claimed; withheld keys prevent
regenerating key-dependent verifier outcomes.

\ifconf\else
  \pdfbookmark[1]{Generative AI Usage}{generative-ai-usage}
  \section*{Generative AI Usage}
  We used Claude (Anthropic) and ChatGPT (OpenAI) for language editing and
  to draft portions of experimental code. The authors reviewed all
  suggestions, ran the experiments end-to-end, and take responsibility for
  all claims, results, and references.
\fi

\begingroup
\ifconf\else\sloppy\fi
\ifconf\else
  \pdfbookmark[1]{References}{references}
\fi
\bibliographystyle{plainurl}
\bibliography{references}

@inproceedings{unmarker_sp25,
  title     = {{UnMarker}: A Universal Attack on Defensive Image Watermarking},
  author    = {Kassis, Andre and Hengartner, Urs},
  booktitle = {2025 IEEE Symposium on Security and Privacy (SP)},
  pages     = {2602--2620},
  publisher = {IEEE},
  address   = {San Francisco, CA, USA},
  year      = {2025},
  doi       = {10.1109/SP61157.2025.00005},
  note      = {arXiv:2405.08363}
}

@inproceedings{watermarkattacker_neurips24,
  title     = {Invisible Image Watermarks Are Provably Removable Using Generative {AI}},
  author    = {Zhao, Xuandong and Zhang, Kexun and Su, Zihao and Vasan, Saastha and Grishchenko, Ilya and Kruegel, Christopher and Vigna, Giovanni and Wang, Yu-Xiang and Li, Lei},
  booktitle = {Advances in Neural Information Processing Systems (NeurIPS)},
  pages     = {8643--8672},
  address   = {Vancouver, BC, Canada},
  year      = {2024},
  doi       = {10.52202/079017-0276},
  note      = {arXiv:2306.01953}
}

@inproceedings{robin_neurips24,
  title     = {{ROBIN}: Robust and Invisible Watermarks for Diffusion Models with Adversarial Optimization},
  author    = {Huang, Huayang and Wu, Yu and Wang, Qian},
  booktitle = {Advances in Neural Information Processing Systems (NeurIPS)},
  pages     = {3937--3963},
  year      = {2024},
  doi       = {10.52202/079017-0129},
  url       = {https://proceedings.neurips.cc/paper_files/paper/2024/hash/073c8584ef86bee26fe9d639ec648e28-Abstract-Conference.html},
  note      = {arXiv:2411.03862}
}

@inproceedings{zodiac_neurips24,
  title     = {Attack-Resilient Image Watermarking Using Stable Diffusion},
  author    = {Zhang, Lijun and Liu, Xiao and {Viros Martin}, Antoni and {Xiong Bearfield}, Cindy and Brun, Yuriy and Guan, Hui},
  booktitle = {Advances in Neural Information Processing Systems (NeurIPS)},
  pages     = {38480--38507},
  year      = {2024},
  doi       = {10.52202/079017-1215},
  url       = {https://proceedings.neurips.cc/paper_files/paper/2024/hash/43d33182360378d5c8e69dd706c24f2f-Abstract-Conference.html},
  note      = {arXiv:2401.04247}
}

@inproceedings{ctrlregen_iclr25,
  title     = {Image Watermarks Are Removable Using Controllable Regeneration from Clean Noise},
  author    = {Liu, Yepeng and Song, Yiren and Ci, Hai and Zhang, Yu and Wang, Haofan and Shou, Mike Zheng and Bu, Yuheng},
  booktitle = {International Conference on Learning Representations (ICLR)},
  address   = {Singapore},
  numpages  = {18},
  year      = {2025},
  url       = {https://openreview.net/forum?id=mDKxlfraAn},
  note      = {arXiv:2410.05470}
}

@inproceedings{nfpa_attack,
  title     = {The Future Unmarked: Watermark Removal in {AI}-Generated Images via Next-Frame Prediction},
  author    = {Qiu, Huming and Wang, Zhaoxiang and Zhang, Mi and Zhang, Xiaohan and You, Xiaoyu and Yang, Min},
  booktitle = {Advances in Neural Information Processing Systems},
  volume    = {38},
  pages     = {138160--138185},
  year      = {2025},
  doi       = {10.52202/085713-4154},
  url       = {https://proceedings.neurips.cc/paper_files/paper/2025/hash/b48405307576689bea14d221ccba73ae-Abstract-Conference.html}
}

@inproceedings{boundary_leakage,
  title     = {Removal Attack and Defense on {AI}-generated Content Latent-based Watermarking},
  author    = {Lee, De Zhang and Fang, Han and Wang, Hanyi and Chang, Ee-Chien},
  booktitle = {Proceedings of the 2025 ACM SIGSAC Conference on Computer and Communications Security},
  pages     = {2174--2188},
  publisher = {Association for Computing Machinery},
  address   = {New York, NY, USA},
  year      = {2025},
  doi       = {10.1145/3719027.3765175},
  note      = {arXiv:2509.11745}
}

@inproceedings{blackbox_forgery_cvpr25,
  title     = {Black-Box Forgery Attacks on Semantic Watermarks for Diffusion Models},
  author    = {M{\"u}ller, Andreas and Lukovnikov, Denis and Thietke, Jonas and Fischer, Asja and Quiring, Erwin},
  booktitle = {2025 IEEE/CVF Conference on Computer Vision and Pattern Recognition (CVPR)},
  pages     = {20937--20946},
  publisher = {IEEE},
  address   = {Nashville, TN, USA},
  year      = {2025},
  url       = {https://openaccess.thecvf.com/content/CVPR2025/html/Muller_Black-Box_Forgery_Attacks_on_Semantic_Watermarks_for_Diffusion_Models_CVPR_2025_paper.html},
  note      = {arXiv:2412.03283}
}

@inproceedings{waves_arxiv24,
  title     = {{WAVES}: Benchmarking the Robustness of Image Watermarks},
  author    = {An, Bang and Ding, Mucong and Rabbani, Tahseen and Agrawal, Aakriti and Xu, Yuancheng and Deng, Chenghao and Zhu, Sicheng and Mohamed, Abdirisak and Wen, Yuxin and Goldstein, Tom and Huang, Furong},
  booktitle = {Proceedings of the 41st International Conference on Machine Learning},
  series    = {Proceedings of Machine Learning Research},
  volume    = {235},
  pages     = {1456--1492},
  publisher = {PMLR},
  address   = {Vienna, Austria},
  year      = {2024},
  url       = {https://proceedings.mlr.press/v235/an24a.html}
}

@inproceedings{hiddenwm_eccv18,
  title     = {{HiDDeN}: Hiding Data with Deep Networks},
  author    = {Zhu, Jiren and Kaplan, Russell and Johnson, Justin and Fei-Fei, Li},
  booktitle = {Proceedings of the European Conference on Computer Vision (ECCV)},
  pages     = {682--697},
  publisher = {Springer},
  address   = {Cham, Switzerland},
  year      = {2018},
  doi       = {10.1007/978-3-030-01267-0_40}
}

@inproceedings{dwtdctsvd,
  title     = {{DWT}-{DCT}-{SVD} based watermarking},
  author    = {Navas, K. A. and Ajay, Mathews Cheriyan and Lekshmi, M. and Archana, Tampy S. and Sasikumar, M.},
  booktitle = {2008 3rd International Conference on Communication Systems Software and Middleware and Workshops (COMSWARE '08)},
  pages     = {271--274},
  publisher = {IEEE},
  address   = {Bangalore, India},
  year      = {2008},
  doi       = {10.1109/COMSWA.2008.4554423}
}

@inproceedings{treering_neurips23,
  title     = {Tree-Rings Watermarks: Invisible Fingerprints for Diffusion Images},
  author    = {Wen, Yuxin and Kirchenbauer, John and Geiping, Jonas and Goldstein, Tom},
  booktitle = {Advances in Neural Information Processing Systems (NeurIPS)},
  address   = {New Orleans, LA, USA},
  numpages  = {17},
  year      = {2023},
  url       = {https://proceedings.neurips.cc/paper_files/paper/2023/hash/b54d1757c190ba20dbc4f9e4a2f54149-Abstract-Conference.html}
}

@inproceedings{gaussianshading_cvpr24,
  title     = {Gaussian Shading: Provable Performance-Lossless Image Watermarking for Diffusion Models},
  author    = {Yang, Zijin and Zeng, Kai and Chen, Kejiang and Fang, Han and Zhang, Weiming and Yu, Nenghai},
  booktitle = {2024 IEEE/CVF Conference on Computer Vision and Pattern Recognition (CVPR)},
  pages     = {12162--12171},
  publisher = {IEEE},
  address   = {Seattle, WA, USA},
  year      = {2024},
  doi       = {10.1109/CVPR52733.2024.01156}
}

@inproceedings{stablesig_iccv23,
  title     = {The Stable Signature: Rooting Watermarks in Latent Diffusion Models},
  author    = {Fernandez, Pierre and Couairon, Guillaume and J{\'e}gou, Herv{\'e} and Douze, Matthijs and Furon, Teddy},
  booktitle = {2023 IEEE/CVF International Conference on Computer Vision (ICCV)},
  pages     = {22409--22420},
  publisher = {IEEE},
  address   = {Paris, France},
  year      = {2023},
  doi       = {10.1109/ICCV51070.2023.02053}
}

@inproceedings{invismark_wacv25,
  title     = {{InvisMark}: Invisible and Robust Watermarking for {AI}-Generated Image Provenance},
  author    = {Xu, Rui and Hu, Mengya and Lei, Deren and Li, Yaxi and Lowe, David and Gorevski, Alex and Wang, Mingyu and Ching, Emily and Deng, Alex},
  booktitle = {Proceedings of the Winter Conference on Applications of Computer Vision (WACV)},
  pages     = {909--918},
  year      = {2025},
  doi       = {10.1109/WACV61041.2025.00098},
  url       = {https://openaccess.thecvf.com/content/WACV2025/html/Xu_InvisMark_Invisible_and_Robust_Watermarking_for_AI-Generated_Image_Provenance_WACV_2025_paper.html},
  note      = {arXiv:2411.07795}
}

@inproceedings{trustmark_iccv25,
  title     = {{TrustMark}: Robust Watermarking and Watermark Removal for Arbitrary Resolution Images},
  author    = {Bui, Tu and Agarwal, Shruti and Collomosse, John},
  booktitle = {Proceedings of the IEEE/CVF International Conference on Computer Vision (ICCV)},
  pages     = {18629--18639},
  year      = {2025},
  url       = {https://openaccess.thecvf.com/content/ICCV2025/html/Bui_TrustMark_Robust_Watermarking_and_Watermark_Removal_for_Arbitrary_Resolution_Images_ICCV_2025_paper.html}
}

@inproceedings{raw_neurips24,
  title     = {{RAW}: A Robust and Agile Plug-and-Play Watermark Framework for {AI}-Generated Images with Provable Guarantees},
  author    = {Xian, Xun and Wang, Ganghua and Bi, Xuan and Srinivasa, Jayanth and Kundu, Ashish and Hong, Mingyi and Ding, Jie},
  booktitle = {Advances in Neural Information Processing Systems (NeurIPS)},
  pages     = {132077--132105},
  publisher = {Neural Information Processing Systems Foundation, Inc.},
  year      = {2024},
  doi       = {10.52202/079017-4198},
  url       = {https://proceedings.neurips.cc/paper_files/paper/2024/hash/ee62ab636066cf45a27246acca9545b7-Abstract-Conference.html},
  note      = {arXiv:2403.18774}
}

@inproceedings{seal_iccv25,
  title     = {{SEAL}: Semantic Aware Image Watermarking},
  author    = {Arabi, Kasra and Witter, R. Teal and Hegde, Chinmay and Cohen, Niv},
  booktitle = {Proceedings of the IEEE/CVF International Conference on Computer Vision (ICCV)},
  pages     = {16196--16205},
  year      = {2025},
  url       = {https://openaccess.thecvf.com/content/ICCV2025/html/Arabi_SEAL_Semantic_Aware_Image_Watermarking_ICCV_2025_paper.html},
  note      = {arXiv:2503.12172}
}

@misc{synthid_image_arxiv25,
  title     = {{SynthID-Image}: Image watermarking at internet scale},
  author    = {Gowal, Sven and Bunel, Rudy and Stimberg, Florian and Stutz, David and Ortiz-Jimenez, Guillermo and Kouridi, Christina and Vecerik, Mel and Hayes, Jamie and Rebuffi, Sylvestre-Alvise and Bernard, Paul and Gamble, Chris and Horv{\'a}th, Mikl{\'o}s Z. and Kaczmarczyck, Fabian and Kaskasoli, Alex and Petrov, Aleksandar and Shumailov, Ilia and Thotakuri, Meghana and Wiles, Olivia and Yung, Jessica and Ahmed, Zahra and Martin, Victor and Rosen, Simon and Sav\v{c}ak, Christopher and Senoner, Armin and Vyas, Nidhi and Kohli, Pushmeet},
  year      = {2025},
  doi       = {10.48550/arXiv.2510.09263},
  url       = {https://arxiv.org/abs/2510.09263},
  note      = {arXiv:2510.09263}
}

@inproceedings{robustwide_eccv24,
  title     = {{Robust-Wide}: Robust Watermarking against Instruction-Driven Image Editing},
  author    = {Hu, Runyi and Zhang, Jie and Xu, Ting and Li, Jiwei and Zhang, Tianwei},
  booktitle = {Proceedings of the European Conference on Computer Vision (ECCV)},
  pages     = {20--37},
  publisher = {Springer},
  address   = {Cham, Switzerland},
  year      = {2024},
  doi       = {10.1007/978-3-031-72670-5_2}
}

@inproceedings{vine_iclr25,
  title     = {Robust Watermarking Using Generative Priors Against Image Editing: From Benchmarking to Advances},
  author    = {Lu, Shilin and Zhou, Zihan and Lu, Jiayou and Zhu, Yuanzhi and Kong, Adams Wai-Kin},
  booktitle = {International Conference on Learning Representations (ICLR)},
  pages     = {83902--83936},
  address   = {Singapore},
  year      = {2025},
  url       = {https://proceedings.iclr.cc/paper_files/paper/2025/hash/d077bc9ea82a2998ca6b2d0158b5ac6e-Abstract-Conference.html},
  note      = {arXiv:2410.18775}
}

@inproceedings{sok_watermarking_aigc,
  title     = {{SoK}: Watermarking for {AI}-Generated Content},
  author    = {Zhao, Xuandong and Gunn, Sam and Christ, Miranda and Fairoze, Jaiden and Fabrega, Andres and Carlini, Nicholas and Garg, Sanjam and Hong, Sanghyun and Nasr, Milad and Tram{\`e}r, Florian and Jha, Somesh and Li, Lei and Wang, Yu-Xiang and Song, Dawn},
  booktitle = {2025 IEEE Symposium on Security and Privacy (SP)},
  pages     = {2621--2639},
  publisher = {IEEE},
  address   = {San Francisco, CA, USA},
  year      = {2025},
  doi       = {10.1109/SP61157.2025.00178}
}

@misc{Aaronson22,
  author       = {Aaronson, Scott},
  title        = {My {AI} Safety Lecture for {UT} Effective Altruism},
  howpublished = {Shtetl-Optimized blog post},
  year         = {2022},
  url          = {https://scottaaronson.blog/?p=6823}
}

@inproceedings{Kirchenbauer_23,
  title     = {A Watermark for Large Language Models},
  author    = {Kirchenbauer, John and Geiping, Jonas and Wen, Yuxin and Katz, Jonathan and Miers, Ian and Goldstein, Tom},
  booktitle = {Proceedings of the 40th International Conference on Machine Learning},
  series    = {Proceedings of Machine Learning Research},
  volume    = {202},
  pages     = {17061--17084},
  publisher = {PMLR},
  address   = {Honolulu, HI, USA},
  year      = {2023},
  url       = {https://proceedings.mlr.press/v202/kirchenbauer23a.html}
}

@inproceedings{C:ChrGun24,
  title     = {Pseudorandom Error-Correcting Codes},
  author    = {Christ, Miranda and Gunn, Sam},
  booktitle = {Advances in Cryptology -- CRYPTO 2024},
  series    = {Lecture Notes in Computer Science},
  volume    = {14925},
  pages     = {325--347},
  publisher = {Springer},
  address   = {Cham, Switzerland},
  year      = {2024},
  doi       = {10.1007/978-3-031-68391-6_10}
}

@inproceedings{gunn2025undetectable,
  title     = {An Undetectable Watermark for Generative Image Models},
  author    = {Gunn, Sam and Zhao, Xuandong and Song, Dawn},
  booktitle = {International Conference on Learning Representations (ICLR)},
  address   = {Singapore},
  numpages  = {26},
  year      = {2025},
  url       = {https://openreview.net/forum?id=jlhBFm7T2J},
  note      = {arXiv:2410.07369}
}

@inproceedings{francati2025codinglimits,
  author    = {Francati, Danilo and Goonatilake, Yevin Nikhel and
               Pawar, Shubham Vivek and Venturi, Daniele and
               Ateniese, Giuseppe},
  title     = {The Coding Limits of Robust Watermarking for
               Generative Models},
  booktitle = {2026 IEEE European Symposium on Security and Privacy
               (EuroS\&P)},
  pages     = {1343--1356},
  publisher = {IEEE},
  address   = {Lisbon, Portugal},
  year      = {2026},
  doi       = {10.1109/EuroSP68448.2026.00087},
  url       = {https://doi.org/10.1109/EuroSP68448.2026.00087},
  note      = {IACR ePrint 2025/1620; arXiv:2509.10577}
}

@inproceedings{zhang2024watermarks,
  title     = {Watermarks in the Sand: Impossibility of Strong Watermarking for Language Models},
  author    = {Zhang, Hanlin and Edelman, Benjamin L. and Francati, Danilo and Venturi, Daniele and Ateniese, Giuseppe and Barak, Boaz},
  booktitle = {Proceedings of the 41st International Conference on Machine Learning},
  series    = {Proceedings of Machine Learning Research},
  volume    = {235},
  pages     = {58851--58880},
  publisher = {PMLR},
  address   = {Vienna, Austria},
  year      = {2024},
  url       = {https://proceedings.mlr.press/v235/zhang24o.html},
  note      = {The ePrint/arXiv version uses the broader title ``Watermarks in the Sand: Impossibility of Strong Watermarking for Generative Models'' and includes corresponding image results.}
}

@inproceedings{marra2019gans,
  title     = {Do {GAN}s Leave Artificial Fingerprints?},
  author    = {Marra, Francesco and Gragnaniello, Diego and Verdoliva, Luisa and Poggi, Giovanni},
  booktitle = {2019 IEEE Conference on Multimedia Information Processing and Retrieval (MIPR)},
  pages     = {506--511},
  publisher = {IEEE},
  address   = {San Jose, CA, USA},
  year      = {2019},
  doi       = {10.1109/MIPR.2019.00103}
}

@inproceedings{frank2020leveraging,
  title     = {Leveraging Frequency Analysis for Deep Fake Image Recognition},
  author    = {Frank, Joel and Eisenhofer, Thorsten and Sch{\"o}nherr, Lea and Fischer, Asja and Kolossa, Dorothea and Holz, Thorsten},
  booktitle = {Proceedings of the 37th International Conference on Machine Learning},
  series    = {Proceedings of Machine Learning Research},
  volume    = {119},
  pages     = {3247--3258},
  publisher = {PMLR},
  address   = {Virtual},
  year      = {2020},
  url       = {https://proceedings.mlr.press/v119/frank20a.html}
}

@inproceedings{durall2020watch,
  title     = {Watch Your Up-Convolution: {CNN}-Based Generative Deep Neural Networks Are Failing to Reproduce Spectral Distributions},
  author    = {Durall, Ricard and Keuper, Margret and Keuper, Janis},
  booktitle = {2020 IEEE/CVF Conference on Computer Vision and Pattern Recognition (CVPR)},
  pages     = {7887--7896},
  publisher = {IEEE},
  address   = {Seattle, WA, USA},
  year      = {2020},
  doi       = {10.1109/CVPR42600.2020.00791}
}

@inproceedings{wang2020cnn,
  title     = {{CNN}-Generated Images Are Surprisingly Easy to Spot\ldots for Now},
  author    = {Wang, Sheng-Yu and Wang, Oliver and Zhang, Richard and Owens, Andrew and Efros, Alexei A.},
  booktitle = {2020 IEEE/CVF Conference on Computer Vision and Pattern Recognition (CVPR)},
  pages     = {8692--8701},
  publisher = {IEEE},
  address   = {Seattle, WA, USA},
  year      = {2020},
  doi       = {10.1109/CVPR42600.2020.00872}
}

@inproceedings{ojha2023universal,
  title     = {Towards Universal Fake Image Detectors That Generalize Across Generative Models},
  author    = {Ojha, Utkarsh and Li, Yuheng and Lee, Yong Jae},
  booktitle = {2023 IEEE/CVF Conference on Computer Vision and Pattern Recognition (CVPR)},
  pages     = {24480--24489},
  publisher = {IEEE},
  address   = {Vancouver, BC, Canada},
  year      = {2023},
  doi       = {10.1109/CVPR52729.2023.02345}
}

@inproceedings{he2016resnet,
  title     = {Deep Residual Learning for Image Recognition},
  author    = {He, Kaiming and Zhang, Xiangyu and Ren, Shaoqing and Sun, Jian},
  booktitle = {2016 IEEE Conference on Computer Vision and Pattern Recognition (CVPR)},
  pages     = {770--778},
  publisher = {IEEE},
  address   = {Las Vegas, NV, USA},
  year      = {2016},
  doi       = {10.1109/CVPR.2016.90}
}

@inproceedings{deng2009imagenet,
  title     = {{ImageNet}: A Large-Scale Hierarchical Image Database},
  author    = {Deng, Jia and Dong, Wei and Socher, Richard and Li, Li-Jia and Li, Kai and Fei-Fei, Li},
  booktitle = {2009 IEEE Conference on Computer Vision and Pattern Recognition},
  pages     = {248--255},
  publisher = {IEEE},
  address   = {Miami, FL, USA},
  year      = {2009},
  doi       = {10.1109/CVPR.2009.5206848}
}

@misc{stable_diffusion_21,
  author       = {{Stability AI}},
  title        = {Stable Diffusion v2.1 and DreamStudio Updates 7-Dec 22},
  year         = {2022},
  url          = {https://stability.ai/news/stablediffusion2-1-release7-dec-2022},
  note         = {Official release post.}
}

@inproceedings{ddim_iclr21,
  title     = {Denoising Diffusion Implicit Models},
  author    = {Song, Jiaming and Meng, Chenlin and Ermon, Stefano},
  booktitle = {International Conference on Learning Representations (ICLR)},
  publisher = {OpenReview.net},
  address   = {Virtual Event, Austria},
  numpages  = {20},
  year      = {2021},
  url       = {https://openreview.net/forum?id=St1giarCHLP},
  note      = {arXiv:2010.02502}
}

@inproceedings{ldm_cvpr22,
  title     = {High-Resolution Image Synthesis with Latent Diffusion Models},
  author    = {Rombach, Robin and Blattmann, Andreas and Lorenz, Dominik and Esser, Patrick and Ommer, Bj{\"o}rn},
  booktitle = {Proceedings of the IEEE/CVF Conference on Computer Vision and Pattern Recognition (CVPR)},
  pages     = {10684--10695},
  publisher = {IEEE},
  address   = {New Orleans, LA, USA},
  year      = {2022},
  doi       = {10.1109/CVPR52688.2022.01042},
  url       = {https://openaccess.thecvf.com/content/CVPR2022/html/Rombach_High-Resolution_Image_Synthesis_With_Latent_Diffusion_Models_CVPR_2022_paper.html}
}

@misc{wu2023human,
  author       = {Wu, Xiaoshi and Hao, Yiming and Sun, Keqiang and Chen, Yixiong and Zhu, Feng and Zhao, Rui and Li, Hongsheng},
  title        = {Human Preference Score v2: A Solid Benchmark for Evaluating Human Preferences of Text-to-Image Synthesis},
  year         = {2023},
  doi          = {10.48550/arXiv.2306.09341},
  url          = {https://arxiv.org/abs/2306.09341},
  note      = {arXiv:2306.09341}
}

@techreport{caltech256_tr,
  title       = {Caltech-256 Object Category Dataset},
  author      = {Griffin, Gregory and Holub, Alex and Perona, Pietro},
  institution = {California Institute of Technology},
  number      = {CNS-TR-2007-001},
  year        = {2007},
  url         = {https://resolver.caltech.edu/CaltechAUTHORS:CNS-TR-2007-001}
}

@misc{u512_kaggle,
  author       = {{rhtsingh}},
  title        = {130k Images ($512 \times 512$) - Universal Image Embeddings},
  year         = {2022},
  howpublished = {Kaggle dataset},
  url          = {https://www.kaggle.com/datasets/rhtsingh/130k-images-512x512-universal-image-embeddings}
}

@misc{abstractart_kaggle,
  author       = {{greg115}},
  year         = {2019},
  title        = {Abstract Art Images},
  howpublished = {Kaggle dataset},
  url          = {https://www.kaggle.com/datasets/greg115/abstract-art}
}

@misc{artmix_kaggle,
  author       = {{sankarmechengg}},
  year         = {2023},
  title        = {Art Images: Clear and Distorted},
  howpublished = {Kaggle dataset},
  url          = {https://www.kaggle.com/datasets/sankarmechengg/art-images-clear-and-distorted}
}

@misc{gustavosta_sd_prompts,
  author       = {{Gustavosta}},
  title        = {Stable-Diffusion-Prompts Dataset},
  year         = {2022},
  howpublished = {Hugging Face dataset card},
  url          = {https://huggingface.co/datasets/Gustavosta/Stable-Diffusion-Prompts}
}

@misc{c2pa_spec,
  author       = {{Coalition for Content Provenance and Authenticity}},
  title        = {Content Credentials: {C2PA} Technical Specification},
  year         = {2024},
  url          = {https://spec.c2pa.org/specifications/specifications/2.1/specs/C2PA_Specification.html}
}

@misc{US23,
  author       = {{The White House}},
  title        = {Executive Order 14110: Safe, Secure, and Trustworthy Development and Use of Artificial Intelligence},
  howpublished = {Federal Register, Vol.~88, No.~210},
  year         = {2023},
  url          = {https://www.govinfo.gov/link/cpd/executiveorder/14110}
}

@misc{EU24,
  author       = {{European Parliament and Council of the European Union}},
  title        = {Regulation (EU) 2024/1689 of the European Parliament and of the Council of 13 June 2024 Laying Down Harmonised Rules on Artificial Intelligence and Amending Regulations (EC) No 300/2008, (EU) No 167/2013, (EU) No 168/2013, (EU) 2018/858, (EU) 2018/1139 and (EU) 2019/2144 and Directives 2014/90/EU, (EU) 2016/797 and (EU) 2020/1828 (Artificial Intelligence Act)},
  howpublished = {Official Journal of the European Union},
  year         = {2024},
  url          = {https://eur-lex.europa.eu/eli/reg/2024/1689/oj/eng}
}

@article{vandermaaten2008tsne,
  title   = {Visualizing Data using {t-SNE}},
  author  = {van der Maaten, Laurens and Hinton, Geoffrey},
  journal = {Journal of Machine Learning Research},
  volume  = {9},
  number  = {86},
  pages   = {2579--2605},
  year    = {2008},
  url     = {https://www.jmlr.org/papers/v9/vandermaaten08a.html}
}

@misc{forensic_cost_arxiv26,
  author    = {Gautier Evennou and Ewa Kijak},
  title     = {The Forensic Cost of Watermark Removal: From Dedicated
               Attacks to Image Editing},
  year      = {2026},
  doi       = {10.48550/arXiv.2604.25491},
  url       = {https://arxiv.org/abs/2604.25491},
  note      = {arXiv:2604.25491v2}
}

@inproceedings{sdedit_iclr22,
  author    = {Meng, Chenlin and He, Yutong and Song, Yang and
               Song, Jiaming and Wu, Jiajun and Zhu, Jun-Yan and
               Ermon, Stefano},
  title     = {{SDE}dit: Guided Image Synthesis and Editing with
               Stochastic Differential Equations},
  booktitle = {International Conference on Learning Representations
               (ICLR)},
  year      = {2022}
}

@inproceedings{zhang2023controlnet,
  author    = {Zhang, Lvmin and Rao, Anyi and Agrawala, Maneesh},
  title     = {Adding Conditional Control to Text-to-Image Diffusion
               Models},
  booktitle = {IEEE/CVF International Conference on Computer Vision
               (ICCV)},
  pages     = {3836--3847},
  year      = {2023}
}

@inproceedings{madry2018towards,
  author    = {Madry, Aleksander and Makelov, Aleksandar and
               Schmidt, Ludwig and Tsipras, Dimitris and Vladu, Adrian},
  title     = {Towards Deep Learning Models Resistant to Adversarial
               Attacks},
  booktitle = {International Conference on Learning Representations
               (ICLR)},
  year      = {2018}
}

@inproceedings{zhang2018lpips,
  author    = {Zhang, Richard and Isola, Phillip and Efros, Alexei A.
               and Shechtman, Eli and Wang, Oliver},
  title     = {The Unreasonable Effectiveness of Deep Features as a
               Perceptual Metric},
  booktitle = {IEEE/CVF Conference on Computer Vision and Pattern
               Recognition (CVPR)},
  pages     = {586--595},
  year      = {2018}
}

@inproceedings{marknull_usenix26,
  author    = {Cao, Jie and Li, Qi and Zhang, Zelin and Wu, Xiaodong and
               Liu, Lingshuang and Li, Xiangman and Ni, Jianbing},
  title     = {{MarkNull}: Model-Agnostic Watermark Removal in
               {AI}-Generated Images via On-Manifold Latent Manipulation},
  booktitle = {35th USENIX Security Symposium (USENIX Security 26)},
  year      = {2026},
  address   = {Baltimore, MD},
  publisher = {USENIX Association},
  month     = aug,
  url       = {https://www.usenix.org/conference/usenixsecurity26/presentation/cao}
}

@inproceedings{marksweep_arxiv26,
  author    = {Cao, Jie and Zhang, Zelin and Li, Qi and Ni, Jianbing},
  title     = {{MarkSweep}: A No-Box Removal Attack on {AI}-Generated
               Image Watermarking via Noise Intensification and
               Frequency-Aware Denoising},
  booktitle = {2026 IEEE International Conference on Acoustics, Speech
               and Signal Processing (ICASSP)},
  pages     = {13932--13936},
  publisher = {IEEE},
  address   = {Barcelona, Spain},
  year      = {2026},
  doi       = {10.1109/ICASSP55912.2026.11460613},
  url       = {https://doi.org/10.1109/ICASSP55912.2026.11460613}
}

@inproceedings{bayar_stamm_ihmms16,
  author    = {Bayar, Belhassen and Stamm, Matthew C.},
  title     = {A Deep Learning Approach to Universal Image Manipulation
               Detection Using a New Convolutional Layer},
  booktitle = {Proceedings of the 4th ACM Workshop on Information
               Hiding and Multimedia Security},
  pages     = {5--10},
  year      = {2016},
  publisher = {ACM},
  doi       = {10.1145/2909827.2930786}
}

@inproceedings{corvi_cvprw23,
  author    = {Riccardo Corvi and Davide Cozzolino and Giovanni Poggi
               and Koki Nagano and Luisa Verdoliva},
  title     = {Intriguing Properties of Synthetic Images: From Generative
               Adversarial Networks to Diffusion Models},
  booktitle = {Proceedings of the IEEE/CVF Conference on Computer Vision
               and Pattern Recognition Workshops},
  pages     = {973--982},
  year      = {2023},
  doi       = {10.1109/CVPRW59228.2023.00104}
}

@inproceedings{dosovitskiy2021image,
  title     = {An Image Is Worth 16x16 Words: Transformers for Image Recognition at Scale},
  author    = {Dosovitskiy, Alexey and Beyer, Lucas and Kolesnikov, Alexander and Weissenborn, Dirk and Zhai, Xiaohua and Unterthiner, Thomas and Dehghani, Mostafa and Minderer, Matthias and Heigold, Georg and Gelly, Sylvain and Uszkoreit, Jakob and Houlsby, Neil},
  booktitle = {International Conference on Learning Representations},
  year      = {2021},
  url       = {https://openreview.net/forum?id=YicbFdNTTy}
}

@article{wang2008perfectly,
  author  = {Wang, Ying and Moulin, Pierre},
  title   = {Perfectly Secure Steganography: Capacity, Error Exponents,
             and Code Constructions},
  journal = {IEEE Transactions on Information Theory},
  volume  = {54},
  number  = {6},
  pages   = {2706--2722},
  year    = {2008},
  doi     = {10.1109/TIT.2008.921684},
  url     = {https://doi.org/10.1109/TIT.2008.921684}
}

@inproceedings{schroederdewitt2023perfectly,
  author    = {Schroeder de Witt, Christian and Sokota, Samuel and
               Kolter, J. Zico and Foerster, Jakob and
               Strohmeier, Martin},
  title     = {Perfectly Secure Steganography Using Minimum Entropy
               Coupling},
  booktitle = {International Conference on Learning Representations
               (ICLR)},
  year      = {2023},
  url       = {https://openreview.net/forum?id=HQ67mj5rJdR},
  note      = {arXiv:2210.14889}
}

@inproceedings{tan2019efficientnet,
  title     = {EfficientNet: Rethinking Model Scaling for Convolutional Neural Networks},
  author    = {Tan, Mingxing and Le, Quoc V.},
  booktitle = {Proceedings of the 36th International Conference on Machine Learning},
  pages     = {6105--6114},
  year      = {2019},
  volume    = {97},
  series    = {Proceedings of Machine Learning Research},
  publisher = {PMLR}
}

@inproceedings{loshchilov2019adamw,
  title     = {Decoupled Weight Decay Regularization},
  author    = {Loshchilov, Ilya and Hutter, Frank},
  booktitle = {International Conference on Learning Representations},
  year      = {2019}
}

@inproceedings{tomasi1998bilateral,
  author    = {Tomasi, Carlo and Manduchi, Roberto},
  title     = {Bilateral Filtering for Gray and Color Images},
  booktitle = {Proceedings of the Sixth International Conference on
               Computer Vision},
  pages     = {839--846},
  publisher = {IEEE Computer Society},
  year      = {1998},
  doi       = {10.1109/ICCV.1998.710815},
  url       = {https://doi.org/10.1109/ICCV.1998.710815}
}

@article{wallace1991jpeg,
  author    = {Wallace, Gregory K.},
  title     = {The {JPEG} Still Picture Compression Standard},
  journal   = {Communications of the ACM},
  volume    = {34},
  number    = {4},
  pages     = {30--44},
  month     = apr,
  publisher = {ACM},
  year      = {1991},
  doi       = {10.1145/103085.103089},
  url       = {https://doi.org/10.1145/103085.103089}
}

@article{dip_watermark_tmlr25,
  author  = {Liang, Hengyue and Li, Taihui and Sun, Ju},
  title   = {A Baseline Method for Removing Invisible Image Watermarks
             Using Deep Image Prior},
  journal = {Transactions on Machine Learning Research},
  year    = {2025},
  url     = {https://openreview.net/forum?id=g85Vxlrq0O}
}

@inproceedings{neurips_removal_challenge25,
  author    = {Shamshad, Fahad and Bakr, Tameem and Shaaban, Yahia
               Salaheldin and Hussein, Noor and Nandakumar, Karthik and
               Lukas, Nils},
  title     = {First-Place Solution to {NeurIPS} 2024 Invisible
               Watermark Removal Challenge},
  booktitle = {1st Workshop on Generative and Protective AI for Content
               Creation (GenProCC), NeurIPS},
  year      = {2025},
  url       = {https://openreview.net/forum?id=SIrUucuYUq}
}

@inproceedings{wang_ldpc_prc_usenix26,
  author    = {Wang, Tianrui and Wang, Anyu and Cong, Tianshuo and
               Ran, Delong and Liu, Jinyuan and Wang, Xiaoyun},
  title     = {Cryptanalysis of {LDPC}-Based Pseudorandom
               Error-Correcting Codes},
  booktitle = {35th USENIX Security Symposium (USENIX Security 26)},
  publisher = {USENIX Association},
  year      = {2026},
  url       = {https://www.usenix.org/conference/usenixsecurity26/presentation/wang-tianrui}
}
\endgroup

\ifconf\else
  \clearpage
\fi
\appendix
\crefalias{section}{appendix}

\ifconf
  \section{Generative AI Usage}
  \label{app:ai-usage}
  We used Claude (Anthropic) and ChatGPT (OpenAI) for language editing and
  to draft portions of experimental code. The authors reviewed all
  suggestions, ran the experiments end-to-end, and take responsibility for
  all claims, results, and references.
\fi

\section{Additional Experimental and Reproducibility Details}
\label{app:experimental-details}

\label{app:training-details}
\paragraph{Training and datasets.}
All six attack-specific detectors follow \Cref{sec:setup:detector}.
Training uses label-smoothed cross-entropy ($0.1$) and
AdamW~\cite{loshchilov2019adamw} (weight decay $10^{-4}$), with cosine
decay after the five-epoch warmup. Runs stop after 50 epochs or seven
epochs without validation-AUROC improvement and use NVIDIA A100~80\,GB
GPUs. No global seed is fixed, and released manifests define the splits. The
scores support metric recomputation and the checkpoints support inference
on those splits; exact retraining is not claimed.

\begin{table}[H]
\centering
\small
\caption{\textbf{Attack-specific dataset sizes.}
Base combines both labels before tampered augmentation; Total includes
tampered examples.}
\label{tab:per-attack-data}
\begin{tabular*}{\columnwidth}{@{\extracolsep{\fill}}llrr@{}}
\toprule
\textbf{Attack} & \textbf{Sources} & \textbf{Base} & \textbf{Total} \\
\midrule
UnMarker       & All five          & 157{,}984 & 197{,}984 \\
WMA            & All five          & 157{,}984 & 197{,}984 \\
CtrlRegen+     & All five          & 157{,}984 & 197{,}984 \\
NFPA           & SD2.1 only        &  49{,}992 &  69{,}992 \\
Boundary Leak. & SD2.1 only        &  19{,}993 &  22{,}993 \\
WiTS           & All five (subset) &  29{,}136 &  32{,}136 \\
\bottomrule
\end{tabular*}
\end{table}

\label{app:pooled-detector}
\paragraph{Pooled detector and representation probe.}
The pooled detector uses a single global progenitor split across all six
attack datasets, keeping each source image and all its derivatives in one
partition. On the corrected held-out mixture, it reaches
AUROC~$0.9976$ and micro-averaged TPR~$97.62\%$ at the
validation-frozen $1\%$-FPR target; per-remover recall ranges from
$95.14\%$ to $99.30\%$. This is performance on the evaluated mixture,
not a universal estimate.

The t-SNE in \Cref{fig:tsne-pooled} projects the earlier checkpoint's
$2{,}048$-dimensional penultimate features with perplexity~50 and PCA
initialization. For quantitative evaluation, a multinomial
logistic-regression probe applied to the corrected checkpoint's
representation distinguishes clean images from the six output sets. Across $20{,}577$ held-out images, fit with $5$-fold GroupKFold on
source identifier, balanced accuracy is $0.971$, versus $0.600$ for
frozen ImageNet-1K features and $0.143$ for random guessing. Because class
sources differ, this is output-set separability, not attack-mechanism
identification or deployment performance.

\section{Validation, Robustness, and Adaptive Evaluation}
\label{app:validation}

\begin{table*}[t]
\centering
\ifconf\small\else\footnotesize\fi
\caption{\textbf{Cross-attack representation controls.}
Arrows compare native evaluation with BMP re-encoding.
$^{a}$Downsample to $256^2$ and restore to $512^2$.
$^{b}$JPEG Q75, resize to 80\%, JPEG Q85, then restore.}
\label{tab:cross-attack-controls}
\setlength{\tabcolsep}{3pt}
\begin{tabular*}{\textwidth}{@{\extracolsep{\fill}}lrrrrrr@{}}
\toprule
\textbf{Control}
& \textbf{UnMarker}
& \textbf{WMA}
& \textbf{CtrlRegen+}
& \textbf{NFPA}
& \textbf{WiTS}
& \textbf{BL.} \\
\midrule
AUROC, Native$\rightarrow$BMP
& 0.9998$\rightarrow$0.9998
& 0.9994$\rightarrow$0.9994
& 0.9999$\rightarrow$0.9999
& 0.9994$\rightarrow$0.9994
& 0.9979$\rightarrow$0.9979
& 0.9995$\rightarrow$0.9995 \\

TPR@0.1\%, Nat$\rightarrow$BMP
& 0.9976$\rightarrow$0.9976
& 0.9462$\rightarrow$0.9462
& 0.9900$\rightarrow$0.9900
& 0.9373$\rightarrow$0.9373
& 0.9623$\rightarrow$0.9623
& 0.9883$\rightarrow$0.9883 \\

Size AUC, Nat.$\rightarrow$BMP
& 0.765$\rightarrow$0.500
& 0.566$\rightarrow$0.500
& 0.553$\rightarrow$0.500
& 0.530$\rightarrow$0.500
& 0.576$\rightarrow$0.500
& 0.580$\rightarrow$0.500 \\

Canonical PNG AUROC
& 0.9998 & 0.9994 & 0.9999 & 0.9994 & 0.9979 & 0.9995 \\

Grayscale AUROC
& 0.987 & 0.988 & 1.000 & 0.998 & 1.000 & 0.997 \\

Downsample AUROC$^{a}$
& 0.986 & 1.000 & 1.000 & 0.997 & 0.997 & 1.000 \\

Social-media AUROC$^{b}$
& 0.766 & 0.999 & 0.995 & 0.995 & 0.995 & 0.999 \\
\bottomrule
\end{tabular*}
\end{table*}

\label{app:cross-attack-controls}
\paragraph{Representation controls.}
BMP and canonical-PNG re-encoding test container, encoding, and
file-size shortcuts; BMP cannot undo artifacts already present in the
pixels. Grayscale, downsampling, and the specified social-media-style
pipeline are transformation stress tests. Their AUROC is at least
$0.985$, except for UnMarker under the latter pipeline ($0.766$).

Metadata audits found no class-specific pattern, and recorded
identifiers produced no exact cross-split duplicates. For NFPA's attack-specific detector dataset, the recorded source identifier
is too coarse to independently reconstruct progenitor-level disjointness.
This limitation does not apply to the separately constructed joint-evaluation cohort, whose outputs and progenitors are disjoint from detector training, validation, and threshold calibration as stated in
\Cref{sec:results:joint}. Family-specific benign controls separately
test generic processing (\Cref{sec:results:validations}).

\label{app:jpeg}
\paragraph{JPEG recompression.}
We apply the same JPEG operation to both classes and read TPRs from each
held-out ROC; this sweep uses an earlier UnMarker checkpoint than
\Cref{tab:unified-detection}.

\begin{table}[H]
\centering
\small
\caption{\textbf{UnMarker detection under JPEG recompression.}}
\label{tab:jpeg-sweep}
\begin{tabular*}{\columnwidth}{@{\extracolsep{\fill}}lrrr@{}}
\toprule
\textbf{Setting}
& \textbf{AUROC}
& \makecell[r]{\textbf{TPR}\\\textbf{@1\%\,FPR}}
& \makecell[r]{\textbf{TPR}\\\textbf{@0.1\%\,FPR}} \\
\midrule
Native  & 0.9994 & 0.9981 & 0.9828 \\
JPG Q90 & 0.9970 & 0.9663 & 0.9192 \\
JPG Q85 & 0.9941 & 0.9326 & 0.2779 \\
JPG Q75 & 0.9771 & 0.7452 & 0.2823 \\
\bottomrule
\end{tabular*}
\end{table}

\label{app:adaptive}
\paragraph{Adaptive evasion and retraining.}
We evaluate detector adaptation on $6{,}000$ CtrlRegen+ outputs.
Starting from each output, we run white-box $\ell_\infty$-bounded
PGD~\cite{madry2018towards} for $60$ steps to minimize the target
detector's score, using $\epsilon=8/255$, step size $1/255$, and an
LPIPS cap of $0.35$. All evasion rates use the detector's $0.5$
decision threshold.

After each round, a fresh ResNet-50 arbiter is initialized from
ImageNet-1K and trained on the resulting outputs. The held-out test set
contains $900$ adaptive and $905$ clean images, with no progenitor
overlap across training, validation, and test. Round one attacks the
original production detector. Each later round attacks the arbiter
trained in the preceding round, but restarts from the original
CtrlRegen+ outputs rather than accumulating edits.
\Cref{tab:adaptive-summary} reports the result of each round.

PGD evades the production detector on $99.97\%$ of outputs in the first
round. We then retrain the arbiter after each round. Against these
retrained arbiters, evasion falls to $0.35\%$, $0.63\%$, and $1.38\%$,
while AUROC remains between $0.9998$ and $1.000$. Median incremental
LPIPS is $0.10$--$0.11$ per round. Each round starts from the original
CtrlRegen+ outputs, so the perturbations do not accumulate. The
cumulative LPIPS of $0.269$ in \Cref{tab:stealth-attempts} comes from a
separate experiment with $3{,}000$ outputs and is measured relative to
the corresponding CtrlRegen+ output.

The PRC verifier was not run on these adaptive outputs. These experiments
therefore evaluate detector evasion and fidelity, not watermark rejection
or complete three-gate success. These finite, CtrlRegen+-only results
show that evading one fixed detector does not establish forensic stealth;
they do not establish robustness to every adaptive attack or transfer to
other removers.

\begin{table}[H]
\centering
\footnotesize
\caption{\textbf{Four-round adaptive retraining on CtrlRegen+.}
Round one attacks the production detector; each later round attacks the
arbiter trained in the preceding round. Every round restarts from the
original CtrlRegen+ outputs. Evasion uses the attacked detector's $0.5$
decision threshold, and Fresh AUROC evaluates a newly trained arbiter on
the current round's held-out outputs.}
\label{tab:adaptive-summary}
\begin{tabular*}{\columnwidth}{@{\extracolsep{\fill}}llrr@{}}
\toprule
\textbf{Round}
& \textbf{Attacked model}
& \textbf{Evasion}
& \textbf{Fresh AUROC} \\
\midrule
1 & Production detector & 99.97\% & 1.0000 \\
2 & Round-1 arbiter     & 0.35\%  & 1.0000 \\
3 & Round-2 arbiter     & 0.63\%  & 1.0000 \\
4 & Round-3 arbiter     & 1.38\%  & 0.9998 \\
\bottomrule
\end{tabular*}
\end{table}

\label{app:stealth-attempts}
\paragraph{Earlier stealth-aware attempts.}
\Cref{tab:stealth-attempts} summarizes four representative ablations.
They measure fidelity and separability, not complete three-gate success. The
battery contains raw ResNet-50, EfficientNet-B0~\cite{tan2019efficientnet}, residual ResNet-50, and DCT ResNet-50.
SDEdit~\cite{sdedit_iclr22} outputs remain separable; whitening moves its
optimized statistic away from the clean target; ControlNet
regeneration~\cite{zhang2023controlnet} fails fidelity and remains
separable.

\begin{table*}[t]
  \centering
  \footnotesize
  \caption{\textbf{Prior stealth-aware ablations.}
Compliance uses cumulative LPIPS to the verified progenitor.
Detectability is the maximum AUROC over the four-detector battery;
the largest matched-clean control AUROC is 0.528. }
  \label{tab:stealth-attempts}
  \setlength{\tabcolsep}{3pt}
  \begin{tabular}{@{}p{0.125\textwidth}
    p{0.045\textwidth}p{0.165\textwidth}p{0.165\textwidth}
    p{\dimexpr0.5\textwidth-8\tabcolsep\relax}@{}}
    \toprule
    \textbf{Method} & \textbf{$n_{\mathrm{run}}$}
      & \textbf{Cumulative LPIPS [compliance]}
      & \textbf{Battery AUROC [95\% CI]}
      & \textbf{Observed failure} \\
    \midrule
    SDEdit \cite{sdedit_iclr22} & 2,834
      & 0.288 [$336/500$; 67.2\%]
      & 0.937 [0.919, 0.954]
      & The edited outputs remain separable; the responsible mechanism
        was not isolated. \\
    $\ell_\infty$ PGD & 3,000   
      & 0.269 [$443/500$; 88.6\%]   
      & 1.000   
      & In this separate run, target-detector evasion reaches $100.0\%$ at the $0.5$ decision threshold, but the held-out battery completely separates the outputs. \\
    Higher-order whitening & 1,014
      & 0.317 [$320/500$; 64.0\%]
      & 1.000
      & The corrected optimized statistic moves from 0.842 to 1.093,
        away from the clean target. \\
    ControlNet regeneration \cite{zhang2023controlnet} & 400
      & 0.553 [$2/400$; 0.5\%]
      & 0.996 [0.988, 1.000]
      & Structural conditioning improves fidelity but does not remove
        finite-test separation. \\
    \bottomrule
  \end{tabular}
\end{table*}

\ifconf\else
  \FloatBarrier
\fi
\section{Supporting Forensic Analysis}
\label{app:forensic-support}

\label{app:visual-comparison}
\paragraph{Illustrative examples.}
\Cref{fig:cross-attack-comparison} gives one illustrative pair for each
attack. These are qualitative; quantitative fidelity appears in
\Cref{tab:unified-detection,tab:joint-outcome}.

\ifconf
  \begin{figure*}[t]
  \centering
  \small
  \setlength{\tabcolsep}{1.5pt}
  \begin{tabular}{cccccc}
  \toprule
  \textbf{UnMarker} & \textbf{WMA} & \textbf{CtrlRegen+}
  & \textbf{NFPA} & \textbf{Boundary Leak.} & \textbf{WiTS} \\
  \midrule
  \includegraphics[width=0.145\linewidth]{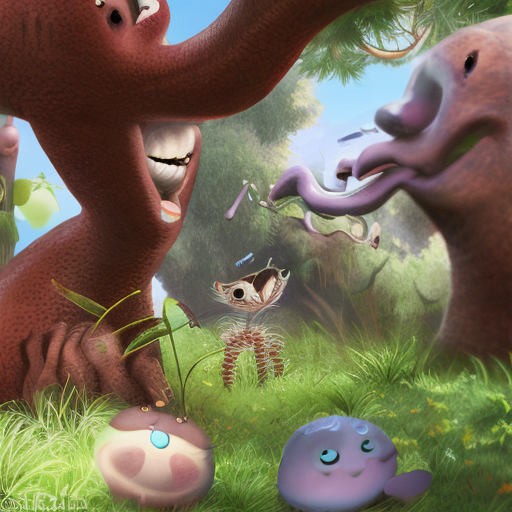} &
  \includegraphics[width=0.145\linewidth]{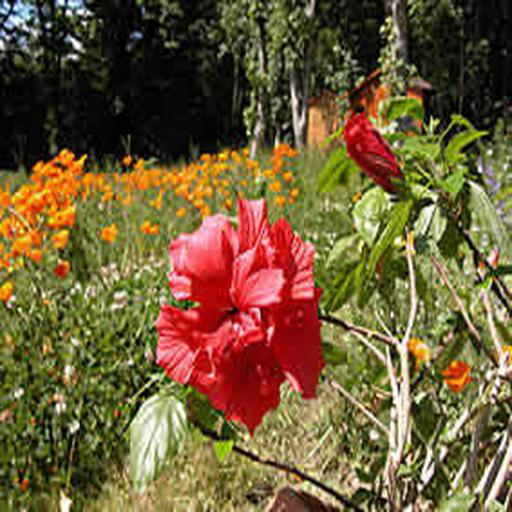} &
  \includegraphics[width=0.145\linewidth]{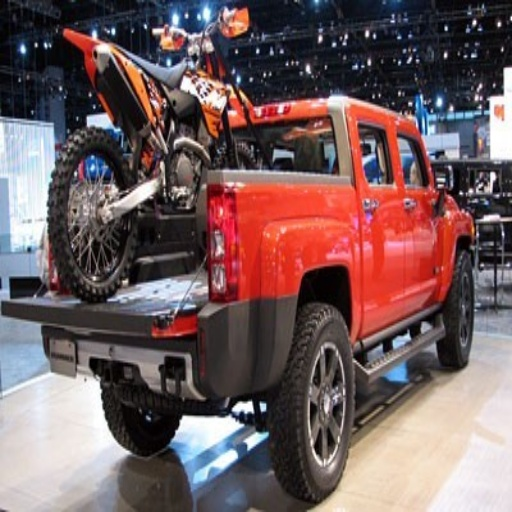} &
  \includegraphics[width=0.145\linewidth]{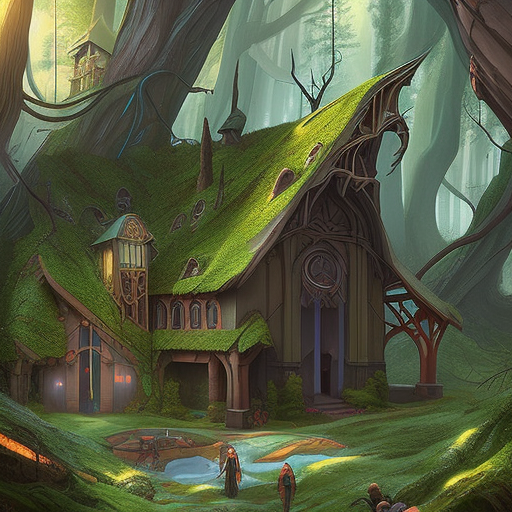} &
  \includegraphics[width=0.145\linewidth]{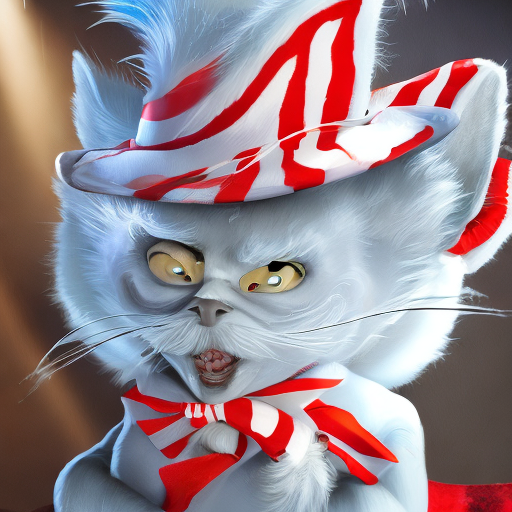} &
  \includegraphics[width=0.145\linewidth]{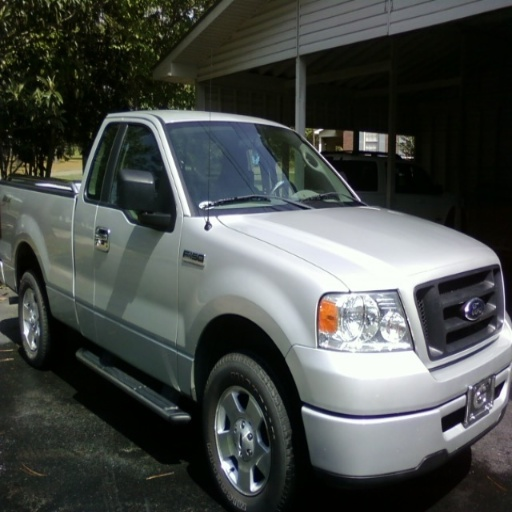} \\
  \includegraphics[width=0.145\linewidth]{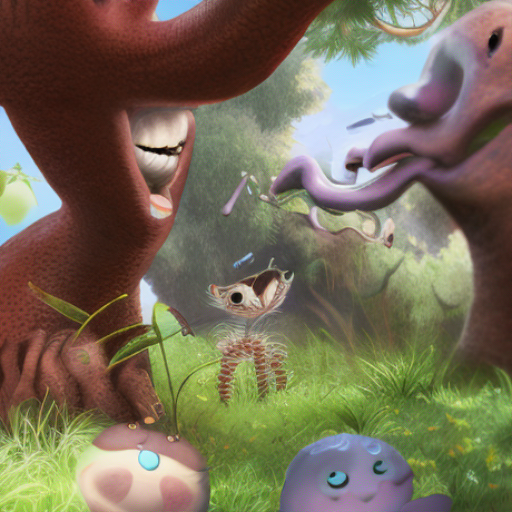} &
  \includegraphics[width=0.145\linewidth]{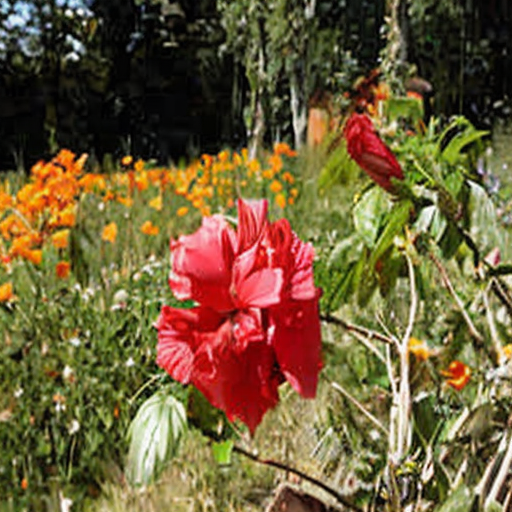} &
  \includegraphics[width=0.145\linewidth]{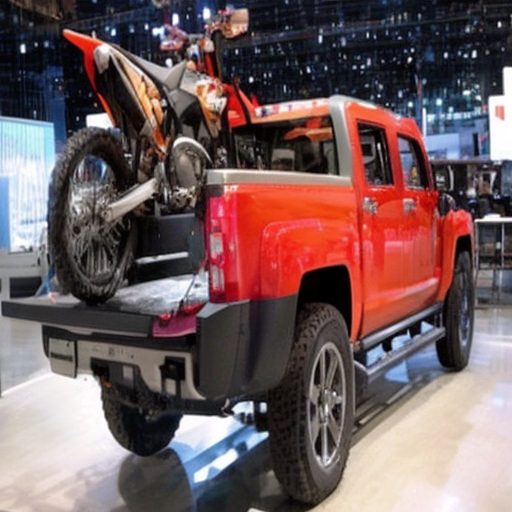} &
  \includegraphics[width=0.145\linewidth]{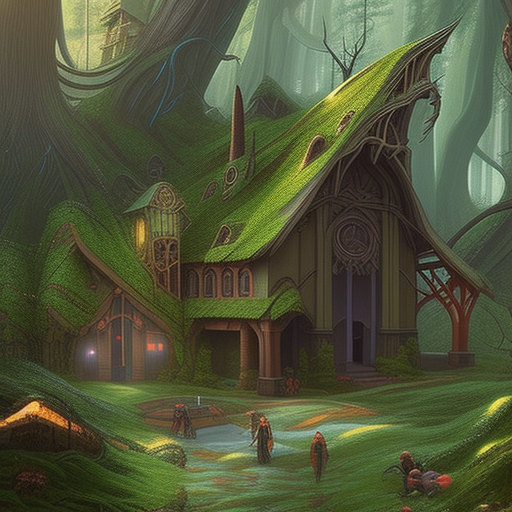} &
  \includegraphics[width=0.145\linewidth]{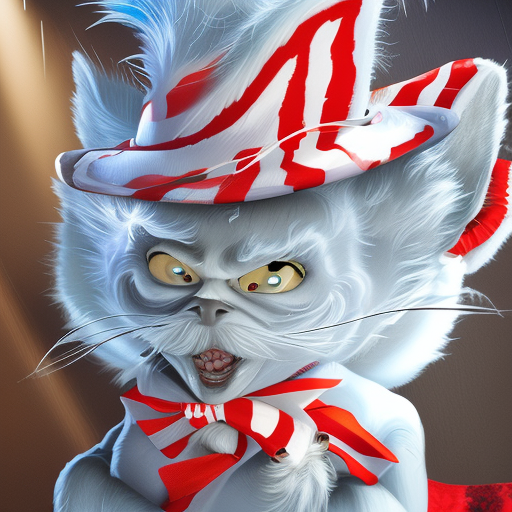} &
  \includegraphics[width=0.145\linewidth]{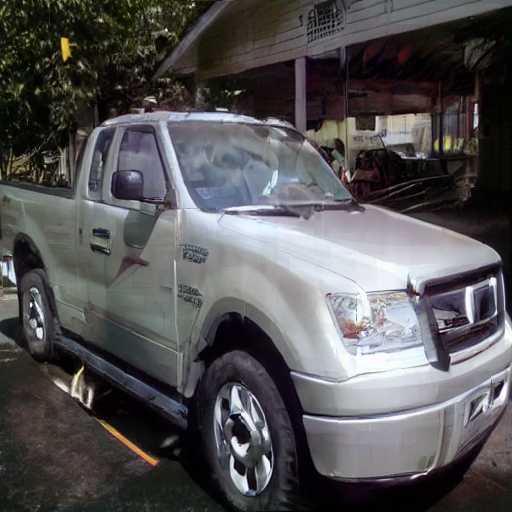} \\
  \bottomrule
  \end{tabular}
\else
  \begin{figure}[H]
  \centering
  \footnotesize
  \setlength{\tabcolsep}{1pt}
  \begin{tabular}{ccc}
  \toprule
  \textbf{UnMarker} & \textbf{WMA} & \textbf{CtrlRegen+} \\
  \midrule
  \includegraphics[width=0.30\linewidth]{figures/clean-ex} &
  \includegraphics[width=0.30\linewidth]{figures/WMA-pre.png} &
  \includegraphics[width=0.30\linewidth]{figures/ctrlregen_pre.png} \\
  \includegraphics[width=0.30\linewidth]{figures/unmarked-ex} &
  \includegraphics[width=0.30\linewidth]{figures/WMA-ex.png} &
  \includegraphics[width=0.30\linewidth]{figures/ctrlregen_ex.png} \\
  \midrule
  \textbf{NFPA} & \textbf{Boundary Leak.} & \textbf{WiTS} \\
  \midrule
  \includegraphics[width=0.30\linewidth]{figures/NFPA-pre.png} &
  \includegraphics[width=0.30\linewidth]{figures/lnp_pre.png} &
  \includegraphics[width=0.30\linewidth]{figures/WiTS_pre.png} \\
  \includegraphics[width=0.30\linewidth]{figures/NFPA-ex.png} &
  \includegraphics[width=0.30\linewidth]{figures/lnp_ex.png} &
  \includegraphics[width=0.30\linewidth]{figures/WiTS_ex.png} \\
  \bottomrule
  \end{tabular}
\fi
\caption{\textbf{Illustrative source--output pairs for the six evaluated
removal attacks.} Top: pre-attack anchor; bottom: attack output.
Boundary Leakage uses an author-released watermarked anchor; the other
pairs follow their respective attack constructions.}
\label{fig:cross-attack-comparison}
\ifconf
  \end{figure*}
\else
  \end{figure}
\fi

\label{app:spectral-methods}
\paragraph{Residual-spectrum methodology.}
Separate analysis manifests pair sources and attacked outputs. From
$512\times512$ uint8 sRGB code values cast to float64 and scaled to
$[0,1]$, we form $r=x_a-x$. Each executable attack uses $N=5{,}000$
pairs, except WiTS ($N=2{,}000$) and Boundary Leakage ($N=4{,}997$).

\ifconf\else
  \begin{figure}[H]
  \centering
  \includegraphics[width=0.78\linewidth]{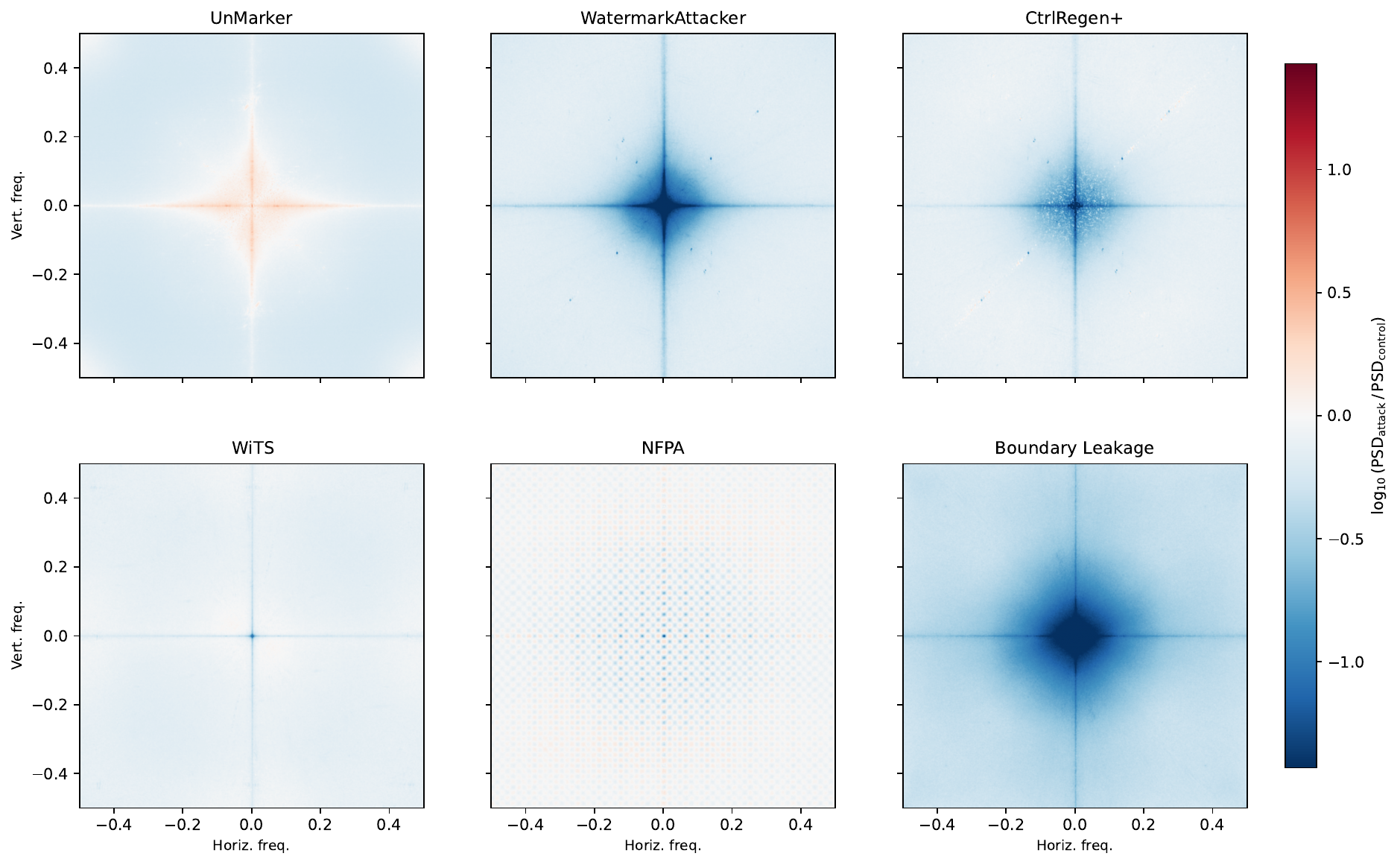}
  \caption{\textbf{2D attack-residual spectra relative to a
  source-matched inter-image reference.}
  Base-10 log-ratio of mean residual PSD to mean reference PSD, with a shared
  color scale: red is positive (more residual power than the reference), blue
  is negative. Boundary Leakage uses a watermarked anchor and is
  not magnitude-comparable with the other attacks.}
  \label{fig:2d-logratio}
  \end{figure}
\fi

Residuals are converted to grayscale, Hann-windowed, transformed by a
centered 2D FFT, and squared. We average PSDs before taking ratios and
apply $(HW)^{-2}$ FFT-size scaling, which does not equalize attack
energies. sRGB is not linearized. Radial profiles use 257 bins over
$[0,0.5]$ cycles/pixel.

The same pipeline applied to unpaired clean-image differences from the
same source mixture gives a source-matched inter-image reference. It is
neither content- nor distortion-matched. We report the base-10
log-ratio of mean residual PSD to mean reference PSD. Its sign indicates
relative residual power, not gain or loss in the attacked image's own
spectrum; the comparison is descriptive, not causal.
\Cref{fig:2d-logratio} shows this statistic before radial averaging.

\ifconf
  \begin{figure}[H]
  \centering
  \includegraphics[width=\linewidth]{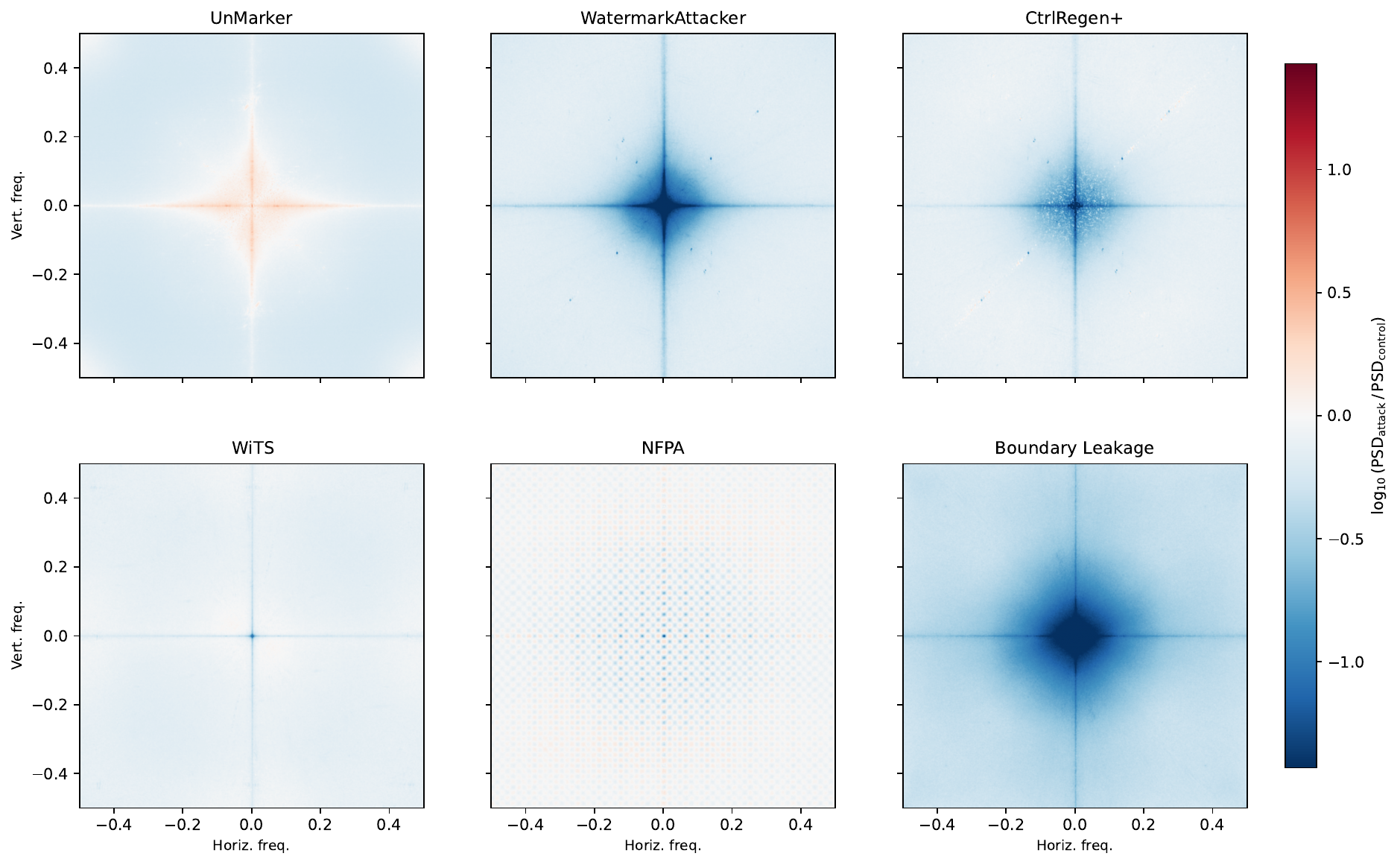}
  \caption{\textbf{2D attack-residual spectra relative to a
  source-matched inter-image reference.}
  Base-10 log-ratio of mean residual PSD to mean reference PSD, with a shared
  color scale: red is positive (more residual power than the reference), blue
  is negative. Boundary Leakage uses a watermarked anchor and is
  not magnitude-comparable with the other attacks.}
  \label{fig:2d-logratio}
  \end{figure}
\fi

\end{document}